\documentclass[english,11pt,3p,authoryear,review]{elsarticle}
\usepackage[T1]{fontenc}
\usepackage[latin9]{inputenc}
\usepackage{calc}
\usepackage{url}
\usepackage{amsmath}
\usepackage{amssymb}
\usepackage{graphicx}
\usepackage{setspace}
\usepackage{longtable,booktabs}
\usepackage{natbib}
\usepackage{hyperref}
\hypersetup{
     colorlinks = true,
     citecolor = blue,
     linkcolor = black
}
\makeatletter

\usepackage{babel}
\newtheorem{lemma}{Lemma}[section]
\newtheorem{corollary}{Corollary}[section]
\newtheorem{proposition}{Proposition}[section]
\usepackage{array}
\newcolumntype{R}[1]{>{\raggedleft\let\newline\\\arraybackslash\hspace{0pt}}m{#1}}
\newcolumntype{L}[1]{>{\raggedright\let\newline\\\arraybackslash\hspace{0pt}}m{#1}}

\newenvironment{proof}[1][Proof]{\begin{trivlist}
\item[\hskip \labelsep {\bfseries #1}]}{\end{trivlist}}

\begin{document}
\title{Generalized multiple depot traveling salesmen problem - polyhedral study and exact algorithm}

\author[sk]{Kaarthik~Sundar\corref{cor1}}

\ead{kaarthiksundar@tamu.edu}

\author[sr]{Sivakumar~Rathinam}

\ead{srathinam@tamu.edu}

\cortext[cor1]{Corresponding author}

\address[sk]{Graduate Student, Dept. of Mechanical Engineering, Texas A$\&$M
University, College Station, TX 77843, USA}
\address[sr]{Assistant Professor, Dept. of Mechanical Engineering, Texas A$\&$M
University, College Station, TX 77843, USA}

\small

\begin{abstract}
The generalized multiple depot traveling salesmen problem (GMDTSP) is a variant of the multiple depot traveling salesmen problem (MDTSP), where each salesman starts at a distinct depot, the targets are partitioned into clusters and at least one target in each cluster is visited by some salesman. The GMDTSP is an NP-hard problem as it generalizes the MDTSP and has practical applications in design of ring networks, vehicle routing, flexible manufacturing scheduling and postal routing. We present an integer programming formulation for the GMDTSP and valid inequalities to strengthen the linear programming relaxation. Furthermore, we present a polyhedral analysis of the convex hull of feasible solutions to the GMDTSP and derive facet-defining inequalities that strengthen the linear programming relaxation of the GMDTSP. All these results are then used to develop a branch-and-cut algorithm to obtain optimal solutions to the problem. The performance of the algorithm is evaluated through extensive computational experiments on several benchmark instances.   
\end{abstract}

\maketitle
\section{Introduction\label{sec:Introduction}}
 The generalized multiple depot travelling salesmen problem (GMDTSP) is an important combinatorial optimization problem that has several practical applications including but not limited to maritime transportation, health-care logistics, survivable telecommunication network design (\cite{Bektas2011}), material flow system design, postbox collection (\cite{Laporte1996}), and routing unmanned vehicles (\cite{Manyam2014, Oberlin2010}). The GMDTSP is formally defined as follows: let $D:=\{d_1,\dots,d_k\}$ denote the set of depots and $T$, the set of targets. We are given a complete undirected graph $G=(V,E)$ with vertex set $V:=T\cup D$ and edge set $E:=\{(i,j):i\in V , j\in T\}$. In addition, a proper partition $C_1,\dots,C_m$ of $T$ is given; these partitions are called \emph{clusters}. For each edge $(i,j)=e \in E$, we associate a non-negative cost $c_e=c_{ij}$. The GMDTSP consists of determining a set of at most $k$ simple cycles such that each cycle starts an ends at a distinct depot, at least one target from each cluster is visited by some cycle and the total cost of the set of cycles is a minimum. The GMDTSP reduces to a multiple depot traveling salesmen problem (MDTSP - \cite{Benavent2013}) when every cluster is a singleton set. The GMDTSP involves two related decisions:
\begin{enumerate}
	\item choosing a subset of targets $S\subseteq T$, such that $|S\cap C_h|\geq 1$ for $h=1,\dots,m$;
	\item solving a MDTSP on the subgraph of $G$ induced by $S\cup D$.
\end{enumerate}

The GMDTSP can be considered either as a generalization of the MDTSP in \cite{Benavent2013} where the targets are partitioned into clusters and at least one target in each cluster has to be visited by some salesman or as a multiple salesmen variant of the symmetric generalized traveling salesman problem (GTSP) in \cite{Fischetti1995, Fischetti1997}. \cite{Benavent2013} and \cite{Fischetti1995} present a polyhedral study of the MDTSP and GTSP polytope respectively, and develop a branch-and-cut algorithm to compute optimal solutions for the respective problem.

This is the first work in the literature that analyzes the facial structure and derives additional valid and facet-defining inequalities for the convex hull of feasible solutions to the GMDTSP. This paper presents a mixed-integer linear programming formulation and develops a branch-and-cut algorithm to solve the problem to optimality. This work generalizes the results of the two aforementioned problems namely the MDTSP (\cite{Benavent2013}) and the GTSP (\cite{Fischetti1995}). 

\subsection{Related work:\label{subsec:litreview}} 

A special case of the GMDTSP with one salesman, the symmetric generalized traveling salesman problem (GTSP), was first introduced by \cite{Henry1969} and \cite{Srivastava1969} in relation to record balancing problems arising in computer design and to the routing of clients through agencies providing various services respectively. Since then, the GTSP has attracted considerable attention in the literature as several variants of the classical traveling salesman problem can be modeled as a GTSP (\cite{Laporte1996,Feillet2005, Oberlin2009, Manyam2014}). \cite{Noon1989} developed a procedure to transform a GTSP to an asymmetric traveling salesman problem and the \cite{Laporte1987} investigated the asymmetric counterpart of the GTSP. Despite most of the aforementioned applications of the GTSP (\cite{Laporte1996}) extending naturally to their multiple depot variant, there are no exact algorithms in the literature to address the GMDTSP.

A related generalization of the GMDTSP can be found in the vehicle routing problem (VRP) literature. VRPs are capacitated counterparts for the TSPs where the vehicles have a limited capacity and each target is associated with a demand that has to be met by the vehicle visiting that target. The multiple VRPs can be classified based on whether the vehicles start from a single depot or from multiple depots. The generalized multiple vehicle routing problem (GVRP) is a capacitated version of the GMDTSP with all the vehicles starting from a single depot. \cite{Bektas2011} present four formulations for the GVRP, compare the linear relaxation solutions for them, and develop a branch-and-cut to optimally solve the problem. In \cite{Laporte1987a}, \citeauthor{Laporte1987a} models the GVRP as a location-routing problem. On the contrary, \cite{Ghiani2000} develop an algorithm to transform the GVRP into a capacitated arc routing problem, which therefore enables one to utilize the available algorithms for the latter to solve the former. In a more recent paper, \cite{Bautista2008} study a special case of the GVRP derived from a waste collection application where each cluster contains at most two vertices. The authors describe a number of heuristic solution procedures, including two constructive heuristics, a local search method and an ant colony heuristic to solve several practical instances. To our knowledge, there are no algorithms in the literature to compute optimal solutions to the generalized multiple depot vehicle routing problem or the GMDTSP. 

The objective of this paper is to develop an integer programming formulation for the GMDTSP, study the facial structure of the GMDTSP polytope and develop a branch-and-cut algorithm to solve the problem to optimality. The rest of the paper is organized as follows: in Sec. \ref{sec:Formulation} we introduce notation and present the integer programming formulation. In Sec. \ref{sec:polyhedral}, the facial structure of the GMDTSP polytope is studied and its relation to the MDTSP polytope (\cite{Benavent2013}) is established. We also introduce a general theorem that allows one to lift any facet of the MDTSP polytope into a facet of the GMDTSP polytope. We further use this result to develop several classes of facet-defining inequalities for the GMDTSP. In the subsequent sections, the formulation is used to develop a branch-and-cut algorithm to obtain optimal solutions. The performance of the algorithm is evaluated through extensive computational experiments on 116 benchmark instances from the GTSP library (\cite{Gutin2010}). 

\section{Problem Formulation\label{sec:Formulation}}
We now present a mathematical formulation for the GMDTSP inspired by models in \cite{Benavent2013} and \cite{Fischetti1995}.
We propose a two-index formulation for the GMDTSP. We associate to each feasible solution $\mathcal{F}$, a vector $\mathbf{x}\in\mathbb{R}^{|E|}$ (a real vector indexed by the elements of $E$) such that the value of the component $x_{e}$ associated with edge $e$ is the number of times $e$ appears in the feasible solution $\mathcal{F}$. Note that for some edges $e\in E$, $x_{e}\in\{0,1,2\}$ \emph{i.e,} we allow the degenerate case where a cycle can only consist of a depot and a target. If $e$ connects two vertices $i$ and $j$, then $(i,j)$ and $e$ will be used interchangeably to denote the same edge. Similarly, associated to $\mathcal{F}$, is also a vector $\mathbf{y}\in\mathbb{R}^{|T|}$, \emph{i.e., } a real vector indexed by the elements of $T$. The value of the component $y_{i}$ associated with a target $i \in T$
is equal to one if the target $i$ is visited by a cycle and zero otherwise. 

For any $S\subset V$, we define $\gamma(S)=\{(i,j)\in E:i,j\in S\}$ and $\delta(S)=\{(i,j)\in E:i\in S,\, j\notin S\}$. If $S=\{i\},$ we simply write $\delta(i)$ instead of $\delta(\{i\})$. We also denote by $C_{h(v)}$ the cluster containing the target $v$ and define $W:=\{v\in T:|C_{h(v)}|=1\}$. Finally, for any $\hat{{E}}\subseteq E$, we define $x(\bar{E})=\sum_{(i,j)\in\bar{E}}x_{ij}$, and for any disjoint subsets $A,B\subseteq V$, $(A:B) = \{(i,j)\in E: i\in A, j\in B\}$ and $x(A:B)=\sum_{e\in(A:B)} x_{ij}$. Using the above notations, the GMDTSP is formulated as a mixed integer
linear program as follows:
\begin{flalign}
 & \text{Minimize}\quad\sum_{e\in E}c_{e}x_{e}\label{eq:obj} &\\
 & \text{subject to}\nonumber &\\
 & x(\delta(i))=2y_{i}\quad\forall i\in T,\label{eq:degree} \\
 & \sum_{i\in C_h} y_{i} \geq 1\quad\forall h\in \{1,\dots,m\},\label{eq:assignment} &\\
 & x(\delta(S))\geq 2y_{i}\quad\forall S\subseteq T,i\in S,\label{eq:sec} &\\
 & x(D':\{j\})+3x_{jk}+x(\{k\}:D\setminus D')\leq2(y_{j}+y_{k})\quad\forall j,k\in T;D'\subset D,\label{eq:4path} &\\
 & x(D':\{j\})+2x(\gamma(S\cup\{j,k\}))+x(\{k\}:D\setminus D')\leq\sum_{v\in S}2\, y_{v}+2(y_{j}+y_{k})-y_i\nonumber &\\
 & \qquad\qquad\qquad\forall i\in S;j,k\in T, j\neq k;S\subseteq T\setminus\{j,k\},S\neq\emptyset;D'\subset D,\label{eq:path} \\
 & x_{e}\in\{0,1\}\quad\forall e\in \gamma(T),\label{eq:xinteger1} &\\
 & x_{e}\in\{0,1,2\}\quad\forall e\in (D:T),\label{eq:xinteger2} &\\
 & y_{i}\in\{0,1\}\quad\forall i\in T.\label{eq:yinteger} 
\end{flalign}
In the above formulation, the constraints in \eqref{eq:degree} ensure the number edges incident on any vertex $i\in T$ is equal to $2$ if and only if target $i$ is visited by a cycle ($y_{i}=1$). The constraints in \eqref{eq:assignment} force at least one target in each cluster to be visited. The constraints in \eqref{eq:sec} are the connectivity or sub-tour elimination constraints. They ensure a feasible solution has no sub-tours of any subset of customers in $T$. The constraints in \eqref{eq:4path} and \eqref{eq:path} are the path elimination constraints. They do not allow for any cycle in a feasible solution to consist of more than one depot. The validity of these constraints is discussed in the subsection \ref{sub:Path-elimination}. Finally, the constraints \eqref{eq:xinteger1}-\eqref{eq:yinteger}
are the integrality restrictions on the $\mathbf{x}$ and $\mathbf{y}$ vectors.
\subsection{Path elimination constraints:\label{sub:Path-elimination} }
The first version of the path elimination constraints was developed in the context of location routing problems \cite{Laporte1986}. \citeauthor{Laporte1986} named these constraints as chain-barring constraints. Authors in \cite{Belenguer2011} and \cite{Benavent2013} use similar path elimination constraints for the location routing and the multiple depot traveling salesmen problems. The version of path elimination constraints used in this article is adapted from \cite{Sundar2014}. Any path that originates from a depot and visits exactly two customers before terminating at another depot is removed by the constraint in \eqref{eq:4path}. The validity of the constraint \eqref{eq:4path} can be easily verified as in \cite{Laporte1986, Sundar2014}. Any other path $d_{1},t_{1},\cdots,t_{p},d_{2}$, where $d_{1},d_{2}\in D$, $t_{1},\cdots,t_{p}\in T$ and $p\geq3$, violates inequality \eqref{eq:path} with $D'=\{d_{1}\},$ $S=\{t_{2},\cdots,t_{p-1}\}$, $j=t_{1}$, $k=t_{p}$ and $i=t_{r}$ where $2\leq r\leq p-1$. The proof of validity of the constraint in Eq. \eqref{eq:path} is discussed as a part of the polyhedral analysis of the polytope of feasible solutions to the GMDTSP in the next section (see proposition \ref{prop:path}).

We note that our formulation allows for a feasible solution with paths connecting two depots and visiting exactly one customer. We refer to such paths as 2-paths. As the formulation allows for two copies of an edge between a depot and a target, 2-paths can be eliminated and therefore there always exists an optimal solution which does not contain any 2-path. In the following subsection, we prove polyhedral results and derive classes of facet-defining inequalities for the model in \eqref{eq:degree}-\eqref{eq:yinteger}.

\section{Polyhedral analysis \label{sec:polyhedral}}
In this section we analyse the facial structure of the GMDTSP polytope while leveraging the results already known for the multiple depot traveling salesmen problem (MDTSP). 

If the number of targets $|T|=n$ and the number of depots $|D|=k$, then the number of $x_e$ variables is $|E| = \binom{n}{2}+nk$ ($\binom{n}{2}$ is the number of edges between the targets and $nk$ is the number of edges between targets and depots). Also the number of $y_i$ variables is $|T|=n$ and hence, the total number of variables used in the problem formulation is $|E|+|T| = \binom n2 + nk + n$. Let $P$ and $Q$ denote the GMDTSP and MDTSP as follows:
\begin{flalign}
P &:= \text{conv}\{(\mathbf{x,y})\in \mathbb{R}^{|E|+|T|}: (\mathbf{x,y}) \text{ is a feasible GMDTSP solution} \} \text{ and }\label{eq:P} \\
Q &:= \{(\mathbf{x,y})\in P: y_v=1 \text{ for all } v\in T \}. \label{eq:Q} 
\end{flalign}
The dimension of the polytope $Q$ was shown to be $\binom n2 + n(k-1)$ in \cite{Benavent2013}. To relate the polytopes $P$ and $Q$, we define an intermediate polytope $P(F)$ as follows:
\begin{flalign}
P(F):= \{(\mathbf{x,y})\in P: y_v = 1 \text{ for all } v \in F\}, \label{eq:PF}
\end{flalign}
where $\emptyset \subseteq F \subseteq T$. Observe that $P(\emptyset) = P$ and $P(T) = Q$. Now, we determine the dimension of the polytope $P(F)$. The number of variables in the equation system for $P(F)$ is $|E|+|T| = \binom n2 + nk + n$. The system also includes $|T|=n$ linear independent equations in \eqref{eq:degree} and variable fixing equations given by $$y_v=1 \text{ for all } v\in F\cup W$$ where,  $W$ is the set of targets that lie in clusters that are singletons (defined in Sec. \ref{sec:Formulation}). The following lemma gives the dimension of $P(F)$.
\begin{lemma} \label{lem:PFdim}
For all $F\subseteq T$, $\operatorname{dim}(P(F)) = \binom n2 + nk - |F\cup W|$.
\end{lemma}
\begin{proof} 
Since the equation system for $P(F)$ has $\binom n2 + nk + n$ variables and $n+|F\cup W|$ linear independent equality constraints, the $\operatorname{dim}(P(F)) \leq \binom n2 + nk - |F\cup W|$. We claim that $P(F)$ contains $\binom n2 + nk - |F\cup W| + 1$ affine independent points. The claim proves $\operatorname{dim}(P(F)) \geq \binom n2 + nk - |F\cup W|$. Hence, the lemma follows. We prove the claim by induction on the cardinality of the set $F$.

For the base case, we have $F = T$ and $P(T)=Q$ where $Q$ is the the MDTSP polytope. Since $\operatorname{dim}(Q) = \binom n2 + nk - n$ (\cite{Benavent2013}), there are $\binom n2 + nk - n + 1$ affine independent points in $Q$. Assume that the claim holds for a set $F_i$ with $|F_i| = i$ and $i>0$, and consider a subset of targets $F_{i-1}$ such that $|F_{i-1}|=i-1$. Let $v$ be any target not in $F_{i-1}$, and define $F_i:= F_{i-1}\cup\{v\}$. The induction hypothesis provides $\binom n2 + nk - |F_i\cup W| + 1$ affine independent points belonging to $P(F_i)$ and hence, to $P(F_{i-1})$ (since $P(F_i)\subseteq P(F_{i-1})$). If $v\in W$, then $|F_{i-1}\cup W| = |F_i\cup W|$ and we are done. Otherwise, $|F_{i-1}\cup W| = |F_i\cup W| - 1$ and we need an additional point on the polytope $P(F_{i-1})$ that is affine independent with the rest of the $\mathcal L = \binom n2 + nk - |F_i\cup W| + 1$ points. All these $\mathcal L$ points satisfy the equation $y_v = 1$. An additional point that is affine independent with the $\mathcal L$ points always exists and is given by any feasible MDTSP solution in the subgraph induced by the set of vertices $(T\cup D)\setminus \{v\}$ because, any feasible MDTSP solution on the set of vertices $(T\cup D)\setminus \{v\}$ satisfies $y_v = 0$. \qed
\end{proof}
\begin{corollary}
$\operatorname{dim}(P) = \binom n2 + nk - |W|$.
\end{corollary}
Lemma \ref{lem:PFdim} indicates that for any given subset $F\subseteq T$ and $v\in F$, either $\operatorname{dim}(P(F\setminus \{v\})) = \operatorname{dim}(P(F))$ (if $v \in W$) or $\operatorname{dim}(P(F\setminus \{v\})) = \operatorname{dim}(P(F)) + 1$ (when $v \notin W$) \emph{i.e.}, the dimension of the polytope $P(F)$ increases by at most one unit when a target is removed from $F$.  Hence, we can lift any facet-defining valid inequality for $P(F)$ to be facet-defining for $P(F\setminus \{v\})$. In the ensuing proposition, we introduce a result based on the sequential lifting for zero-one programs (\cite{Padberg1975}) which we will use to lift facets of $Q$ into facets of $P$. The proposition generalizes a similar result in \cite{Fischetti1995} used to lift facets of the travelling salesman problem to facets of GTSP. 
\begin{proposition} \label{prop:lifting}
Suppose that for any $F\subseteq T$ and $u\in F$, $$\sum_{e\in E} \alpha_e x_e + \sum_{v\in T} \beta_v (1-y_v) \geq \eta$$ is any facet-defining inequality for $P(F)$. Then the lifted inequality  $$\sum_{e\in E} \alpha_e x_e + \sum_{v\in T\setminus \{u\}} \beta_v (1-y_v) + \bar{\beta}_u(1-y_u) \geq \eta$$ is valid and facet-defining for $P(F\setminus \{u\})$, where $\bar{\beta}_u$ takes an arbitrary value when $u \in W$ and $$\bar{\beta}_u = \eta - \min \left\{\sum_{e\in E} \alpha_e x_e + \sum_{v\in T\setminus \{u\}} \beta_v (1-y_v): (\mathbf{x,y}) \in P(F\setminus \{u\}), y_u=0 \right\} $$ when $u \notin W$. Note that the statement can be trivially modified to deal with ``$\leq$'' inequalities.
\end{proposition}
\begin{proof}
The proof follows from the sequential lifting theorem in \cite{Padberg1975}. \qed
\end{proof}
Proposition \ref{prop:lifting} is used to derive facet-defining inequalities for the GMDTSP polytope $P$ by lifting the facet-defining inequalities for the MDTSP polytope $Q$ in \cite{Benavent2013}. For a given lifting sequence of the set of targets $T$, say $\{v_1,\dots,v_n\}$, the procedure is iteratively applied to derive a facet of $P(\{v_{t+1},\dots,v_n\})$ from a facet of $P(\{v_{t},\dots,v_n\})$ for $t=1,\dots,n$. Different lifting sequences produce different facets; hence the name, \emph{sequence dependent} lifting. In the rest of the section, we use the lifting procedure to check if the constraints in \eqref{eq:degree}-\eqref{eq:yinteger} are facet-defining and derive additional facet-defining inequalities for the GMDTSP polytope.
\begin{proposition} \label{prop:trivial}
The following results hold for the GMDTSP polytope $P$:
\begin{enumerate}
	\item $x_e \geq 0$ defines a facet for every $e \in E$ if $|T|\geq 4$,
	\item $x_e\leq 1$ defines a facet if and only if $e\in \gamma(W)$ and $|T|\geq 3$,
	\item $x_e \leq 2$ does not define a facet for any $e\in (D:T)$,
	\item $y_i\geq 0$ does not define a facet for any $i\in T$,
	\item $y_i\leq 1$ defines a facet if and only if $i\notin W$, and
	\item $\sum_{i\in C_h} y_{i} \geq 1$ does not define a facet for any $h\in \{1,\dots,m\}$. 
\end{enumerate}
\end{proposition}
\begin{proof} We use the facet-defining results of the MDTSP polytope (\cite{Benavent2013}) in conjunction with Proposition \ref{prop:lifting} to prove (1)--(3).
\begin{enumerate}
\item Observe that for every $e \in E$, $x_e\geq 0$ defines a facet of the MDTSP polytope $Q$ if $|T| \geq 4$. Now for any lifting sequence, Proposition \ref{prop:lifting} produces $\bar{\beta}_v = 0$ for all $v\in T$ and the result follows.
\item Suppose that $e = (i,j)$. If $i,j\in W$ and $|T|\geq 3$, then the claim follows from the forthcoming Proposition \ref{prop:gsec} by choosing $S = \{i,j\}$. Otherwise if $e=(i,j)\in \gamma(T)$, then $x_e \leq 1$ is dominated by $x_e \leq y_i$ if $i \notin W$ and $x_e \leq y_j$ if $j \notin W$.
\item Let $e = (d,i)$ where $d\in D, i\in T$. $x_e \leq 2$ defines a face of the MDTSP polytope $Q$. Hence neither of the lifted versions of the inequality \emph{i.e.}, $x_e \leq 2$ (if $i \in W$) or $x_e \leq 2y_i$ (if $i\notin W$) defines a facet of $P$.
\item The inequality $y_i \geq \frac 12 x_e$ for $e \in \delta(i)$ dominates $y_i \geq 0$. Hence, $y_i\geq 0$ does not define a facet for any $i\in T$.
\item Observe that the valid inequality $y_i\leq 1$ induces a face, $P(\{i\}) = \{(\mathbf{x,y})\in P: y_i = 1\}$ of $P$. From the Lemma \ref{lem:PFdim}, $\operatorname{dim}(P(\{i\})) = \operatorname{dim}(P) - 1$ if and only if $i\notin W$. Hence, $y_i \leq 1$ is facet-defining for $P$ if and only if $i\notin W$. When $i\in W$, the inequality defines an improper face. 
\item The constraint $\sum_{i\in C_h} y_{i} \geq 1$ can be reduced to $\sum_{e\in \delta(C_h)} x_e + 2\sum_{e\in \gamma(C_h)} x_e \geq 2$ using the degree constraints in \eqref{eq:degree}.  When $\gamma(C_h) \neq \emptyset$, the constraint $\sum_{e\in \delta(C_h)} x_e + 2\sum_{e\in \gamma(C_h)} x_e \geq 2$ is dominated by $\sum_{e\in \delta(C_h)} x_e \geq 2$. When $\gamma(C_h) = \emptyset$ (\emph{i.e.,} $|C_h| = 1$), the constraint $\sum_{e\in \delta(C_h)} x_e = 2$ is satisfied by any feasible solution in $P$ and hence in this case, it is an improper face. Therefore, $\sum_{i\in C_h} y_{i} \geq 1$ does not define a facet for any $h\in \{1,\dots,m\}$. \qed
\end{enumerate}
\end{proof}
In the next proposition, we prove that the sub-tour elimination constraints in Eq. \eqref{eq:sec} define facets of $P$. To do so, we apply the lifting procedure in Proposition \ref{prop:lifting} to the MDTSP sub-tour elimination constraints $$x(\delta(S)) \geq 2 \text{ for all } S\subseteq T.$$ In the process, we derive alternate versions of the sub-tour elimination constraints in Eq. \eqref{eq:sec} which we will refer to as the generalized sub-tour elimination constraints (GSEC). To begin with, we observe that sub-tour elimination constraints given above define facets of the MDTSP poytope $Q$ when $|T|\geq 3$ (see \cite{Benavent2013}). 
\begin{proposition} \label{prop:gsec}
Let $S\subseteq T$ and $|T|\geq 3$. Then the following \emph{Generalized Sub-tour Elimination Constraint} (GSEC) is valid and facet-defining for $P$: $$x(\delta(S)) + \bar{\beta}_i (1-y_i) \geq 2 \text{ for } i\in S,$$ where $$\bar{\beta}_i = \begin{cases} 2 & \text{ if } \mu(S)=0, \\ 0 & \text{ otherwise};\end{cases}$$ $\mu(S)$ is defined as $\mu(S) = |\{h:C_h \subseteq S\}|$.
\end{proposition}
\begin{proof}
We first observe that the inequality $x(\delta(S)) \geq 2$ with $S\subseteq T$ and $|T|\geq 3$ defines a facet for the MDTSP polytope. We lift this inequality using the lifting procedure in Proposition \ref{prop:lifting}. Let $\{v_1,\dots,v_n\}$ be any lifting sequence of the set of targets such that $v_n = i$. The lifting coefficients $\bar{\beta}_{v_t}$ are computed iteratively for $t=1,\dots,n$. For $t=1,\dots,n-1$, it is trivial to see that $\bar{\beta}_{v_t} = 0$. Hence, $x(\delta(S)) \geq 2$ defines a facet of $P(\{v_n\})$. As to $\bar{\beta}_{v_n}$, we compute its value by performing the lifting procedure again and obtain a facet of $P$. We have $$\bar{\beta}_{v_n} = 2 - \min \left\{x(\delta(S)): (\mathbf{x,y}) \in P, \text{ and } y_{v_n} = 0\right\}.$$ Solving for $\bar{\beta}_{v_n}$ using the above equation, we obtain $\bar{\beta}_{v_n} = 2$ if a feasible GMDTSP solution visiting no target in $S$ exists (\emph{i.e.}, no $C_h \subseteq S$ exists) and $\bar{\beta}_{v_n} = 0$ otherwise. \qed
\end{proof}
In summary, the Proposition \ref{prop:gsec} results in the following facet-defining inequalities of $P$: suppose $S\subseteq T$ with $|T|\geq 3$. Then,
\begin{flalign}
x(\delta(S)) &\geq 2 \text{ for } \mu(S)\neq 0 \text{ and } \label{eq:sec1} \\
x(\delta(S)) &\geq 2y_i \text{ for } \mu(S) = 0,\, i\in S. \label{eq:sec2} 
\end{flalign}
Note that the inequality $x(\delta(S)) \geq 2y_i$ is valid for any $S \subseteq T$. It is facet-defining for $P$ only when $\mu(S) \neq 0$. When $\mu(S)=0$ it does not define a facet of $P$ as it is dominated by Eq. \eqref{eq:sec1}. Using the degree constraints in Eq. \eqref{eq:degree}, the above GSECs can rewritten as 
\begin{flalign}
x(\gamma(S)) &\leq \sum_{v\in S} y_v - 1 \text{ for } \mu(S)\neq 0 \text{ and } \label{eq:sec1a} \\
x(\gamma(S)) &\leq \sum_{v\in S\setminus\{i\}} y_v \text{ for } \mu(S) = 0,\, i\in S. \label{eq:sec2a} 
\end{flalign}
In the forthcoming two propositions, we prove that the path elimination constraints in Eq. \eqref{eq:4path} and \eqref{eq:path} are facet-defining of $P$ using Proposition \ref{prop:lifting}. The corresponding path elimination constraints for the MDTSP polytope $Q$ are as follows: suppose that $j,k \in T$, $D'\subset D$ with $D'\neq \emptyset$, then
\begin{flalign}
 x(D':\{j\})+3x_{jk}+x(\{k\}:D\setminus D')&\leq 4 \label{eq:4pathmdtsp} \\
 x(D':\{j\})+2x(\gamma(S\cup\{j,k\}))+x(\{k\}:D\setminus D')&\leq 2|S|+3 \text{ for } S\subseteq T\setminus\{j,k\},S\neq\emptyset \label{eq:pathmdtsp} 
\end{flalign}
We remark that Eq. \eqref{eq:4pathmdtsp} and \eqref{eq:pathmdtsp} define facets for the MDTSP polytope $Q$ (see \cite{Benavent2013}).
\begin{proposition} \label{prop:4path}
Suppose $j,k\in T$ and $D' \subset D$ with $D'\neq \emptyset$. Then the following path elimination constraint is valid and facet-defining for $P$: $$x(D':\{j\})+3x_{jk}+x(\{k\}:D\setminus D') + \bar{\beta}_j (1-y_j) + \bar{\beta}_k (1-y_k) \leq 4$$ where $\bar{\beta}_j = \bar{\beta}_k = 2$.
\end{proposition}
\begin{proof}
Let $\{v_1,\dots,v_n\}$ be any lifting sequence of the set of targets such that $v_{n-1} = j$ and $v_n = k$. The lifting coefficients are iteratively computed for $t=1,2,\dots,n$. Coefficients $\bar{\beta}_v$ for $v\in\{v_1,\dots,v_{n-2}\}$ are easily computed (tight GMDTSP solution is depicted in Fig. \ref{fig:4path}(a), showing that the value of $\bar{\beta}_v$ cannot be increased without producing a violated inequality).
\begin{figure}
\centering
	\begin{minipage}[t]{.45\textwidth} \centering
		\includegraphics[scale=1]{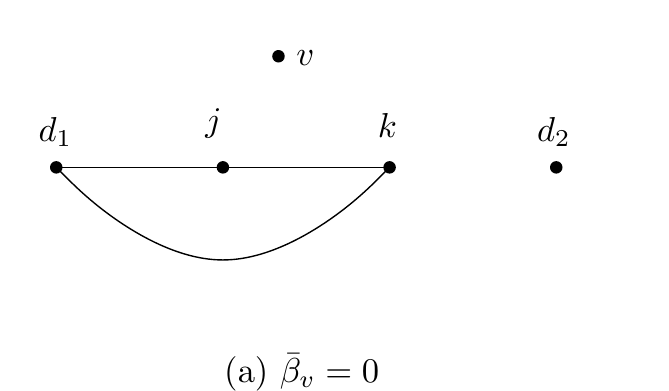}
	\end{minipage}\hfill
	\begin{minipage}[t]{.45\textwidth} \centering
		\includegraphics[scale=1]{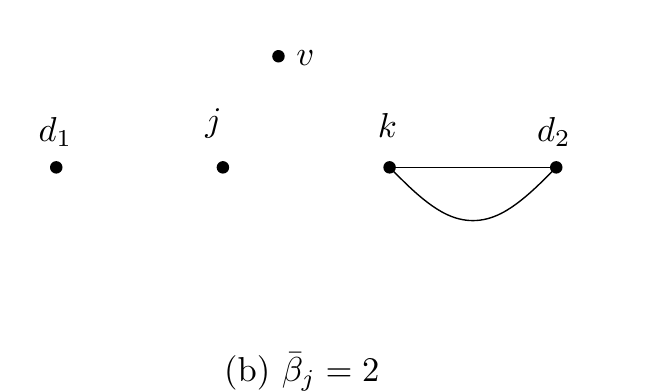}
	\end{minipage}
\caption{Tight feasible solutions for proof of Proposition \ref{prop:4path}. The vertices $d_1$ and $d_2$ are depots and $j,k,$ and $v$ are targets.} \label{fig:4path}
\end{figure}
Similarly for $t=n-1$ \emph{i.e.}, $v_t = j$, the correctness of the coefficient $\bar{\beta}_j = 2$ can be checked with the help of Fig. \ref{fig:4path}(b). Analogously, we obtain $\bar{\beta}_{k} = 2$. \qed
\end{proof}
The inequality in Proposition \ref{prop:4path} can be rewritten as $x(D':\{j\})+3x_{jk}+x(\{k\}:D\setminus D')\leq2(y_{j}+y_{k})$ which is the path elimination constraint in Eq. \eqref{eq:4path}. We have proved that this inequality is valid and defines a facet of $P$. 
\begin{figure}[htbp]
\centering
	\begin{minipage}[t]{.5\textwidth} \centering
		\includegraphics[scale=1]{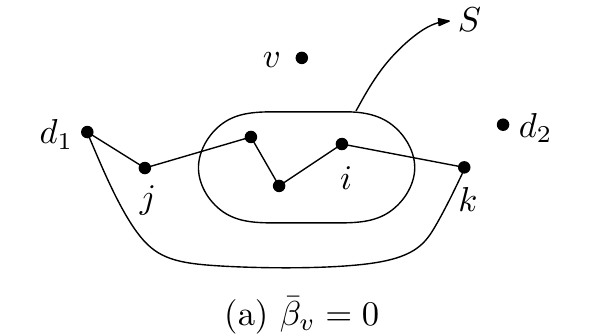}
	\end{minipage}\hfill
	\begin{minipage}[t]{.5\textwidth} \centering
		\includegraphics[scale=1]{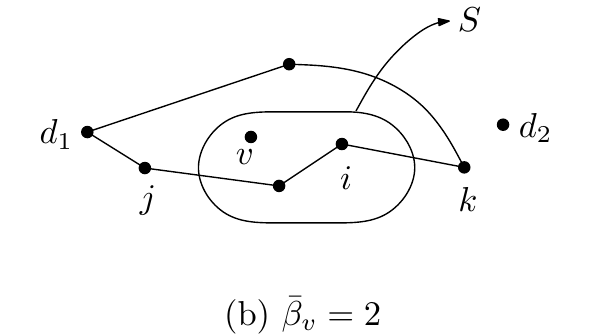}
	\end{minipage} \\
	\begin{minipage}[t]{.5\textwidth} \centering
		\includegraphics[scale=1]{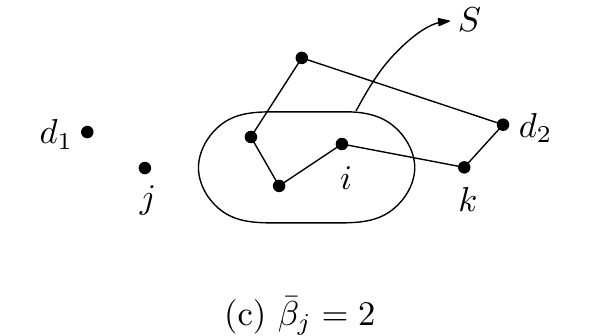}
	\end{minipage}\hfill
	\begin{minipage}[t]{.5\textwidth} \centering
		\includegraphics[scale=1]{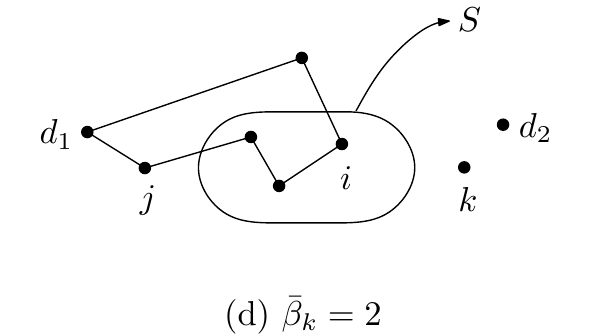}
	\end{minipage}\\
	\begin{minipage}[t]{.5\textwidth} \centering
		\includegraphics[scale=1]{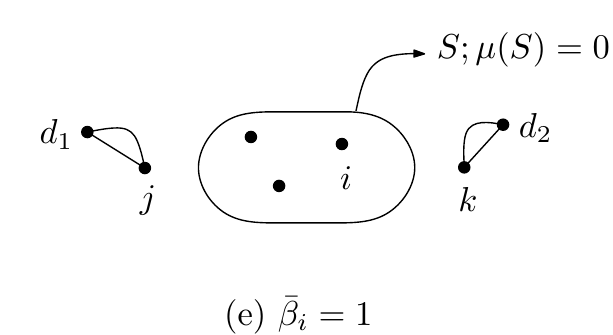}
	\end{minipage}\hfill
	\begin{minipage}[t]{.5\textwidth} \centering
		\includegraphics[scale=1]{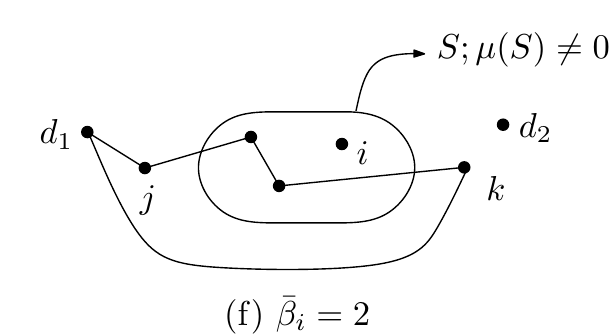}
	\end{minipage}
\caption{Tight feasible solutions for proof of Proposition \ref{prop:path}.} 
\label{fig:path}
\end{figure}
\begin{proposition} \label{prop:path}
Let $j,k\in T$, $D' \subset D$,  $S\subseteq T\setminus \{j,k\}$ and $i\in S$ such that $D'\neq \emptyset$ and $S\neq \emptyset$. Also let $\bar S = S\cup \{j,k\}$. Then the following \emph{Generalized Path Elimination Constraint} (PSEC) is valid and facet-defining for $P$: $$x(D':\{j\})+2x(\gamma(\bar S))+x(\{k\}:D\setminus D') + \sum_{v\in T} \bar{\beta}_v (1-y_v) \leq 2|S|+3$$ where $$\bar{\beta}_v = \begin{cases} 0 & \text{ if } v \in T\setminus{\bar S}, \\ 2 & \text{ if } v \in \bar S\setminus\{i\}, \\ 1 & \text{ if } v = i \text{ and } \mu(S) = 0, \\ 2 & \text{ if } v = i \text{ and } \mu(S) \neq 0;\end{cases} $$ $\mu(S)$ is defined as $\mu(S) = |\{h:C_h \subseteq S\}|$.
\end{proposition}
\begin{proof}
Consider any lifting sequence of the the set of targets $\{v_1,\dots,v_n\}$ such that each target in the set $S\setminus \{i\}$ follows all the targets in the set $|T\setminus \bar S|$ and $v_{n-2}=j$, $v_{n-1}=k$ and $v_n = i$. The coefficients $\bar\beta_v = 0$ for $v\in T\setminus \bar S$ and $\bar\beta_v = 2$ for $v\in S\setminus \{i\}$ are trivial to compute (tight GMDTSP solution is depicted in Fig. \ref{fig:path}(a) and \ref{fig:path}(b) respectively, showing that the value of $\bar{\beta}_v$ cannot be increased without producing a violated inequality). The correctness of coefficients $\bar \beta_j=2$ and $\bar \beta_k=2$ can be checked with the help of Fig. \ref{fig:path}(c) and \ref{fig:path}(d), respectively. 

It remains to compute the value of coefficient $\bar\beta_i$. For computing $\bar\beta_i$, we have to take into account for the possibility of a GMDTSP solution not visiting any target in the set $S$. This can happen when $\mu(S) = 0$. In this case, we obtain $\bar \beta_i = 1$; see Fig. \ref{fig:path}(e). Likewise, when $\mu(S) \neq 0$, any GMDTSP solution has to have at least two edges in $\delta(S)$. This leads to $\bar \beta_i = 2$; tight GMDTSP solution is shown in Fig. \ref{fig:path}(f). \qed
\end{proof}
In summary, the Proposition \ref{prop:path} results in the following facet-defining inequalities of $P$: suppose $j,k\in T$, $D' \subset D$,  $S\subseteq T\setminus \{j,k\}$, $\bar S = S\cup \{j,k\}$ and $i\in S$ such that $D'\neq \emptyset$ and $S\neq \emptyset$, then
\begin{flalign}
x(D':\{j\})+2x(\gamma(\bar S))+x(\{k\}:D\setminus D')&\leq \sum_{v\in \bar S}2y_v - y_i \text{ for } \mu(S)= 0 \text{ and } \label{eq:pec1} \\
x(D':\{j\})+2x(\gamma(\bar S))+x(\{k\}:D\setminus D')&\leq \sum_{v\in \bar S}2y_v - 1 \text{ for } \mu(S)\neq 0. \label{eq:pec2} 
\end{flalign}
We note that the above PSECs can be rewritten in cut-set form as
\begin{flalign}
x(\delta(\bar S)) &\geq x(D':\{j\})+x(\{k\}:D\setminus D')+ y_i \text{ for } \mu(S)= 0 \text{ and } \label{eq:pec1a} \\
x(\delta(\bar S)) &\geq x(D':\{j\})+x(\{k\}:D\setminus D')+ 1 \text{ for } \mu(S)\neq 0. \label{eq:pec2a} 
\end{flalign}
As we will see in the forthcoming section, the GPECs in the above form are more amicable for developing separation algorithms. Next, we examine the comb inequalities that are valid and facet-defining for the MDTSP polytope. These inequalities were initially introduced for the TSP in \cite{Chvatal1973}. These inequalities were extended and proved to be facet-defining for the MDTSP polytope in \cite{Benavent2013}. We define a comb inequality using a comb, which is a family $C = (H,\mathcal T_1, \mathcal T_2,\dots, \mathcal T_t)$ of $t+1$ subsets of the targets; $t$ is an odd number and $t\geq 3$. The subset $H$ is called the handle and the subsets $\mathcal T_1,\dots,\mathcal T_t$ are called teeth. The handle and teeth satisfy the following conditions: 
\begin{enumerate}[i.]
\item $H\cap \mathcal T_i \neq \emptyset \quad \forall i=1,\dots,t,$
\item $\mathcal T_i\setminus H \neq \emptyset \quad \forall i=1,\dots,t,$
\item $\mathcal T_i\cap \mathcal T_j = \emptyset \quad 1\leq i\leq j \leq t$.
\end{enumerate}
The conditions i. and ii. indicate that every tooth $T_i$ intersects the handle $H$ and the condition iii. indicates that no two teeth intersect. We define the size of $C$ as $\sigma(C):= |H| + \sum_{i=1}^t|\mathcal T_i| - \frac{3t+1}{2}$. Then the comb inequality associated with $C$ is given by 
\begin{flalign}
x(\gamma(H)) + \sum_{i=1}^t x(\gamma(\mathcal T_i)) &\leq \sigma(C) \label{eq:comba}
\end{flalign}
The inequality in Eq. \ref{eq:comba} is valid and facet-defining for the MDTSP (see \cite{Benavent2013}). A special case of the comb inequality, called \emph{2-matching} inequality is obtained when $|\mathcal T_i| = 2$ for $i=1,\dots,t$. In the case of a 2-matching inequality, the size of the comb is $\sigma(C) = |H| + \frac{t+1}{2}$. We apply the lifting procedure in Proposition \ref{prop:lifting} to the inequality in \eqref{eq:comba} and obtain facet-defining inequality for the GMDTSP. The following proposition is adapted from \cite{Fischetti1995}; the proof of the proposition is omitted as it is similar to the proof of the corresponding theorem for GTSP in \cite{Fischetti1995}.
\begin{proposition} \label{prop:comb}
Suppose $\mu(S) = |\{h:C_h \subseteq S\}|$ for $S\subseteq T$ and let $C = (H,\mathcal T_1,\dots,\mathcal T_t)$ be a comb. For $i=1,\dots,t$, let $a_i$ be any target in $\mathcal T_i\cap H$ if $\mu(\mathcal T_i\cap H)=0$; $a_i=0$ (a dummy value) otherwise; and let $b_i$ be any target in $\mathcal T_i\setminus H$ if $\mu(\mathcal T_i\setminus H) = 0$; $b_i=0$ otherwise. Then the following comb inequality is valid and facet-defining for the GMDTSP polytope $P$:
\begin{flalign}
x(\gamma(H)) + \sum_{i=1}^t x(\gamma(\mathcal T_i)) + \sum_{v\in T} \bar{\beta}_v (1-y_v) &\leq \sigma(C), \label{eq:combb}
\end{flalign}
where $\bar \beta_v = 0$ for all $v\in T\setminus (H\cup \mathcal T_1 \cup \dots\cup \mathcal T_t)$, $\bar \beta_v = 1$ for all $v\in H\setminus (\mathcal T_1 \cup \dots\cup \mathcal T_t)$ and for $i=1,\dots,t$:
\begin{flalign*}
\bar \beta_v &= 2 \quad \text{for $v\in \mathcal T_i \cap H, v\neq a_i$;} \\
\bar \beta_{a_i} &= 1 \quad \text{if $a_i \neq 0$;} \\
\bar \beta_v &= 1 \quad \text{for $v\in \mathcal T_i \setminus H, v\neq b_i$;} \\
\bar \beta_{b_i} &= 0 \quad \text{if $b_i \neq 0$}.
\end{flalign*}
\end{proposition}
\begin{proof}
See \cite{Fischetti1995}. \qed
\end{proof}

\subsection{Additional valid inequalities specific to multiple depot problems \label{subsec:Tcomb}} In this section, we will examine a special type of comb inequality called the T-comb inequalities. The T-comb inequalities were introduced in \cite{Benavent2013} and proved to be valid and facet-defining for the MDTSP polytope. These inequalities are specific to problems involving multiple depots and hence, are important for the GMDTSP. A T-comb inequality $C$ is defined by an handle $H$ and teeth $\mathcal T_1, \dots, \mathcal T_t$ such that the following conditions are satisfied:
\begin{enumerate}[i.]
\item $H\cap \mathcal T_i \neq \emptyset \quad \forall i=1,\dots,t,$
\item $\mathcal T_i\setminus H \neq \emptyset \quad \forall i=1,\dots,t,$
\item $\mathcal T_i\cap \mathcal T_j = \emptyset \quad 1\leq i\leq j \leq t,$
\item $\mathcal T_i \cap D \neq \emptyset \quad \forall i=1,\dots,t,$
\item $H \subset T,$
\item $H\setminus \cup_{i=1}^t \mathcal T_i \neq \emptyset,$
\item $D\setminus \cup_{i=1}^t \mathcal T_i \neq \emptyset$.
\end{enumerate} 
The difference between the T-comb inequalities and the comb inequalities defined in Eq. \eqref{eq:comba} is that, the number of teeth are allowed to be even ($t\geq 1$) and each teeth must contain a depot. The comb size in this case is given by $\sigma(C)= |H| + \sum_{i=1}^t|\mathcal T_i| - (t+1)$. In this paper, we will only examine the T-comb inequalities with $|\mathcal T_i| = 2$ for every $i\in\{1,\dots,t\}$; the size of the comb in this case reduces to $\sigma(C) = |H| + t - 1$ and the corresponding T-comb inequality is given by 
\begin{flalign}
x(\gamma(H)) + \sum_{i=1}^t x(\gamma(\mathcal T_i)) &\leq |H| + t - 1, \label{eq:Tcomba}
\end{flalign}
The inequality in Eq. \eqref{eq:Tcomba} is valid and facet-defining for the MDTSP when $t\geq 2$. Again, we apply the lifting procedure in Proposition \ref{prop:lifting} to the inequality in \eqref{eq:Tcomba} and obtain facet-defining inequality for the GMDTSP. 
\begin{proposition} \label{prop:Tcomb}
Let $C = (H,\mathcal T_1,\dots,\mathcal T_t)$ be a T-comb with $|\mathcal T_i| = 2$ for every $i\in\{1,\dots,t\}$ and $t\geq 2$. Also suppose $|H\setminus \cup_i \mathcal T_i| > 1$ $($the proposition can be trivially extended to the case where $|H\setminus \cup_i \mathcal T_i| = 1)$. Let $\bar a$ be any target in $H\setminus \cup_i \mathcal T_i$. Then the following T-comb inequality is valid and facet-defining for the GMDTSP polytope $P$:
\begin{flalign}
x(\gamma(H)) + \sum_{i=1}^t x(\gamma(\mathcal T_i)) + \sum_{v\in T} \bar{\beta}_v (1-y_v) &\leq |H| + t - 1, \label{eq:Tcombb}
\end{flalign}
where $\bar \beta_v = 0$ for all $v\in T\setminus (H\cup \mathcal T_1 \cup \dots\cup \mathcal T_t)$, $\bar \beta_v = 1$ for all $v\in H\setminus (\mathcal T_1 \cup \dots\cup \mathcal T_t \cup \{\bar a\})$, $\bar \beta_{\bar a} = 0$, and $\bar \beta_v = 2$ for all $v\in \mathcal T_i \cap H,i=1,\dots,t$.
\end{proposition}
\begin{proof}
\begin{figure}[htbp]
\centering
	\begin{minipage}[t]{.5\textwidth} \centering
		\includegraphics[scale=1]{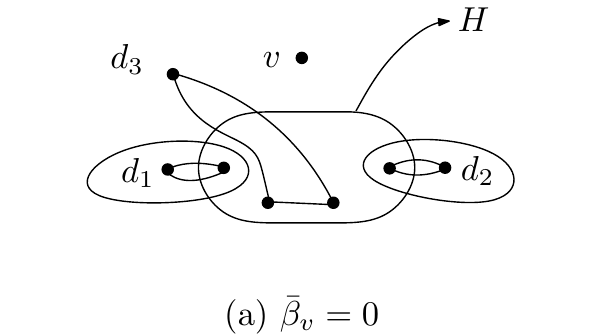}
	\end{minipage}\hfill
	\begin{minipage}[t]{.5\textwidth} \centering
		\includegraphics[scale=1]{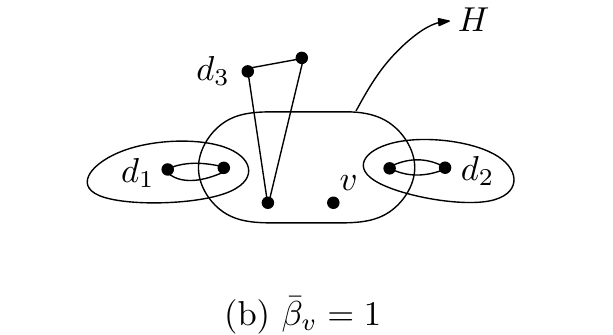}
	\end{minipage} \\
	\begin{minipage}[t]{.5\textwidth} \centering
		\includegraphics[scale=1]{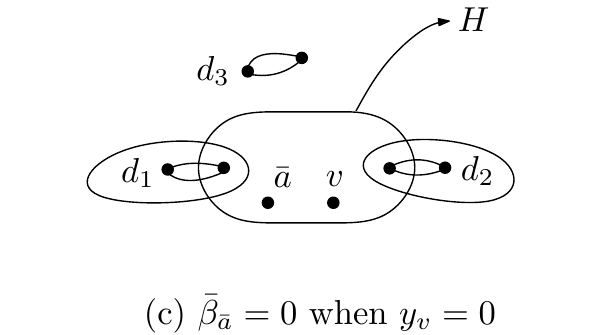}
	\end{minipage}\hfill
	\begin{minipage}[t]{.5\textwidth} \centering
		\includegraphics[scale=1]{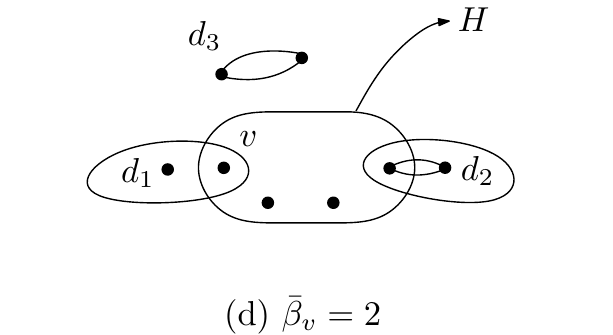}
	\end{minipage}
	\caption{Tight feasible solutions for proof of Proposition \ref{prop:Tcomb}.} 
\label{fig:Tcomb}
\end{figure}
Consider any lifting sequence for the set of targets $T$ in the following order: (i) targets in the set $T\setminus (H\cup \mathcal T_1 \cup \dots\cup \mathcal T_t)$, (ii) $v\in H\setminus (\mathcal T_1 \cup \dots\cup \mathcal T_t \cup \{\bar a\})$, (iii) $\bar a$, and (iv) $v\in \mathcal T_i \cap H, i=1,\dots,t$. The lifting coefficients $\bar\beta_v = 0$ and $\bar\beta_v = 1$ for the sets in (i) and (ii) respectively, are trivial to compute (tight feasible GMDTSP solutions are depicted in Fig. \ref{fig:Tcomb}(a) and \ref{fig:Tcomb}(b), respectively). Similarly, tight feasible GMDTSP solutions for the cases where $\bar\beta_{\bar a} = 0$ and $\bar\beta_v = 2$ (cases (iii) and (iv)) are shown in Fig. \ref{fig:Tcomb}(c) and \ref{fig:Tcomb}(d), respectively. \qed
\end{proof}
In the above proposition, for the case when $|H\setminus \cup_i \mathcal T_i| = 1$, the facet-defining inequality is given by 
\begin{flalign}
x(\gamma(H)) + \sum_{i=1}^t x(\gamma(\mathcal T_i)) &\leq \sum_{i=1}^t \sum_{v\in H\cap \mathcal T_i} 2y_v. \label{eq:Tcombc}
\end{flalign}

\section{Separation algorithms\label{sec:separation}}
In this section, we discuss the algorithms that are used to find violated families of all the valid inequalities introduced in Sec. \ref{sec:polyhedral}. We denote by $G^* = (V^*,E^*)$ the \emph{support graph} associated with a given fractional solution $(\mathbf{x}^*, \mathbf{y}^*) \in \mathbb{R}^{|E|\cup|T|}$ \emph{i.e.,} $G^*$ is a capacitated undirected graph with vertex set $V^* := \{i\in T: y^*_i >0\}\cup D$ and $E^*:=\{e\in E:x_e^* > 0\}$ with edge capacities $x^*_e$ for each edge $e\in E^*$.

\subsection{Separation of generalized sub-tour elimination constraints in Eq. \eqref{eq:sec1} and Eq. \eqref{eq:sec2}:}  \label{subsec:sec}
We first develop a separation algorithm for constraints in Eq. \eqref{eq:sec2}:  $x(\delta(S)) \geq 2y_i$ for $\mu(S) = 0, i\in S$ and $S\subseteq T$. Given a fractional solution $(\mathbf{x}^*, \mathbf{y}^*)$, the most violated constraint of the form \eqref{eq:sec2} can be obtained by computing a minimum capacity cut $(S,V^*\setminus S)$ with $i\in S$ and $D\subseteq V^*\setminus S$ on the graph $G^*$. The minimum capacity cut can be obtained by computing a maximum flow from $i$ to $t$, where $t$ is an additional vertex connected with each depot in the set $D$ through an edge having very large capacity. The algorithm is repeated for every target $i \in T \cap V^*$ and the target set $S$ obtained during each run of the algorithm defines a violated inequality if the capacity of the cut is strictly less than $2y^*_i$. This procedure can be implemented in $O(|T|^4)$ time.

Now we consider the constraint in Eq. \eqref{eq:sec1}: $x(\delta(S)) \geq 2$ for $\mu(S) \neq 0$ and $S\subseteq T$. Given a fractional solution $(\mathbf{x}^*, \mathbf{y}^*)$, the most violated inequality \eqref{eq:sec1} in this case is obtained by computing a minimum capacity cut $(S,V^*\setminus S)$ with a cluster $C_h \subseteq S$ and $D\subseteq V^*\setminus S$ on the graph $G^*$. This is in turn achieved by computing a maximum $s-t$ flow on $G^*$, where $s$ and $t$ are additional vertices connected with each $j\in C_h$ and each $d \in D$ respectively through an edge having very large capacity. The algorithm is repeated for every cluster $C_h$ and the set $S$ obtained on each run of the algorithm defines a violated inequality if the capacity of the cut is strictly less than $2$. The time complexity of this procedure is $O(m|T|^3)$, where $m$ is the number of clusters. 

We remark that the violated inequality of the form \eqref{eq:sec2} using the above algorithm, is not necessarily facet-defining as the set $S$ computed using the algorithm might have $\mu(S) \neq 0$. When this happens, we reject the inequality in favour of its dominating and facet-defining inequality in Eq. \eqref{eq:sec1}. 

\subsection{Separation of path elimination constraints in Eq. \eqref{eq:4path}, \eqref{eq:pec1}, and \eqref{eq:pec2}:} \label{subsec:pec}
We first discuss the procedure to separate violated constraints in Eq. \eqref{eq:4path}. Consider every pair of targets $j,k \in V^*\cap T$. We rewrite the constraint in \eqref{eq:4path} as $x(D':\{j\}) + x(\{k\}:D\setminus D') \leq 2(y_k + y_j)-3x_{jk}$. Given $j,k$ and a fractional solution $(\mathbf{x}^*, \mathbf{y}^*)$, the RHS of the above inequality is a constant and is equal to $2(y_k^* + y_j^*) - 3x_{jk}^*$. We observe that the LHS of the inequality is maximized when $D'=\{d\in D:x^*_{jd} \geq x^*_{kd}\}$. Furthermore, when $D'=\emptyset$ or $D' = D$, no path constraint in Eq. \eqref{eq:4path} is violated for the given pair of vertices. With $D' = \{d\in D: x^*_{jd} \geq x^*_{kd}\}$, if $x^*(D':\{j\}) + x^*(\{k\}:D\setminus D')$ is strictly greater than $2(y_k^* + y_j^*) - 3x_{jk}^*$, the path constraint in Eq. \eqref{eq:4path} is violated for the pair of vertices $j,k$ and the subset of depots $D'$. This procedure can be implemented in $O(|T|^2)$.

For constraints in Eq. \eqref{eq:pec1} and \eqref{eq:pec2}, we present two separation algorithms that are very similar to the algorithms presented in Sec. \ref{subsec:sec}. We will use the equivalent constraints in Eq. \eqref{eq:pec1a} and \eqref{eq:pec2a} to develop the algorithms. We first consider the path elimination constraint in Eq. \eqref{eq:pec2a}. Given $j,k$ and a fractional solution $(\mathbf{x}^*, \mathbf{y}^*)$, we first compute $D'$ to maximize $x^*(D':\{j\}) + x^*(\{k\}:D\setminus D') := \mathcal L$. Now, the most violated constraint of the form \eqref{eq:pec2a} can be obtained by computing a minimum capacity cut $(\bar S, V^*\setminus \bar S)$ with $j,k\in \bar S$, a cluster $C_h \subseteq \bar S\setminus\{j,k\}$ and $D\subseteq V^*\setminus \bar S$. This algorithm is repeated for every target $j,k\in T$ and cluster $C_h$ such that $j,k \notin C_h$ and the target set $S = \bar S\setminus \{j,k\}$ obtained during each run of the algorithm defines a violated inequality if the capacity of the cut is strictly less than $\mathcal L + 1$. The time complexity of this algorithm is $O(m|T|^4)$. Similarly, the most violated constraint of the form \eqref{eq:pec1a} can be obtained by computing a minimum capacity cut $(\bar S, V^*\setminus \bar S)$, with $i,j,k \in \bar S$ and $D\subseteq V^*\setminus \bar S$ on the graph $G^*$. This algorithm is repeated for very triplet of targets in $V^*$ and the set $S = \bar S\setminus \{j,k\}$ defines a violated inequality if the capacity of the cut is strictly less than $\mathcal L+y_i^*$. The time complexity of the algorithm is $O(|T|^5)$. 

Similar to the separation of the sub-tour elimination constraints, we remark that the violated inequality of the form \eqref{eq:pec1a}, computed using the above algorithm  is not necessarily facet-defining as the set $S$ might have $\mu(S) \neq 0$. When this happens, we reject the inequality in favour of its dominating and facet-defining inequality in Eq. \eqref{eq:pec2a}.

\subsection{Separation of comb inequalities in Eq. \eqref{eq:combb}} \label{subsec:comb} 
For the comb-inequalities in Eq. \eqref{eq:combb}, we use the separation procedures discussed in \cite{Fischetti1997}. We first consider the special case of the comb inequalities with $|\mathcal T_i|=2$ for $i=1,\dots,t$ \emph{i.e.,} the 2-matching inequalities. Using a construction similar to the one proposed in \cite{Padberg1982} for the $b$-matching problem, the separation problem for the 2-matching inequalities can be transformed into a minimum capacity off cut problem; hence this separation problem is exactly solvable in polynomial time. But this procedure is computationally intensive, and so we use the following heuristic proposed by \cite{Grotschel1985}. Given a fractional solution $(\mathbf{x}^*, \mathbf{y}^*)$, the heuristic considers a graph $\bar G = (\bar V, \bar E)$ where $\bar V = V^* \cap T$ and $\bar E = \{e: 0<x_e^*<1\}$. Then, we consider each connected component $H$ of $\bar G$ as a handle of a possibly violated 2-matching inequality whose two-vertex teeth correspond to edges $e \in \delta(H)$ with $x_e^* = 1$. We reject the inequality if the number of teeth is even. The time complexity of this algorithm is $O(|\bar V| + |\bar E|)$. As for the comb inequalities, we apply the same procedure after shrinking each cluster into a single supernode.  

\subsection{Separation of T-comb inequalities in Eq. \eqref{eq:Tcombb} and \eqref{eq:Tcombc}} \label{subsec:Tcomb}
We present a separation heuristic similar to the one used in \cite{Benavent2013} to identify violated T-comb inequalities of the form Eq. \eqref{eq:Tcombb} and \eqref{eq:Tcombc}. We first build a set of teeth, each containing a distinct depot according to the following procedure: a tooth $\mathcal T_i$ is built by starting with a set containing a depot $d\in D$; a target $v\in T$ is added to $\mathcal T_i$ such that $x(\delta(\mathcal T_i))$ is a minimum. Then, for every subset of this set of teeth such that: (i) they are pairwise disjoint, (ii) belong to the same connected component of the support graph $G^* = (V^*,E^*)$, and (iii) do not together contain all the targets of that connected component, an appropriate handle $H$ is built as follows: assume $H$ is the set of all the targets in the connected component and remove the targets in $H\setminus (\mathcal T_i\cup \dots \cup \mathcal T_t)$ sequentially. Every time a target is removed, the T-comb inequality of the appropriate form is checked for violation. The time complexity of this algorithm is $O(|T|)$.

\section{Branch-and-cut algorithm \label{sec:bandc}}
In this section, we describe important implementation details of the branch-and-cut algorithm for the GMDTSP. The algorithm is implemented within a CPLEX 12.4 framework using the CPLEX callback functions \cite{cplex124}. The callback functions in CPLEX enable the user to completely customize the branch-and-cut algorithm embedded into CPLEX, including the choice of node to explore in the enumeration tree, the choice of branching variable, the separation and the addition of user-defined cutting planes and the application of heuristic methods.

The lower bound at the root node of the enumeration tree is computed by solving the LP relaxation of the formulation in Sec. \ref{sec:Formulation} that is further strengthened using the cutting planes described in Sec. \ref{sec:polyhedral}. The initial linear program consisted of all constraints in \eqref{eq:obj}-\eqref{eq:yinteger}, except \eqref{eq:sec}, \eqref{eq:4path} and \eqref{eq:path}. For a given LP solution, we identify violated inequalities using the separation procedures detailed in Sec. \ref{sec:separation} in the following order: (i) sub-tour elimination constraints in Eq. \eqref{eq:sec1}, (ii) sub-tour elimination constraints in Eq. \eqref{eq:sec2} (iii) path elimination constraints in Eq. \eqref{eq:4path}, \eqref{eq:pec1} and, \eqref{eq:pec2}, (iv) generalized comb constraints in Eq. \eqref{eq:combb}, and (v) T-comb constraints in Eq. \eqref{eq:Tcombb} and \eqref{eq:Tcombc}. This order of adding the constraints to the formulation was chosen after performing extensive computational experiments. Furthermore, we disabled the separation of all the cuts embedded into the CPLEX framework because enabling these cuts increased the average computation time for the instances. Once the new cuts generated using these separation procedures were added to the linear program, the tighter linear program was resolved. This procedure was iterated until either of the following conditions was satisfied: (i) no violated constraints could be generated by the separation procedures, (ii) the current lower bound of the enumeration tree was greater or equal to the current upper bound. If no constraints are generated in the separation phase, we create subproblems by branching on a fractional variable. First, we select a fractional $y_{i}$ variable, based on the \emph{strong branching} rule (\cite{Achterberg2005}). If all these variables are integers, then we select a fractional $x_{e}$ variable using the same rule. As for the node-selection rule, we used the best-first policy for all our computations,\emph{i.e.}, select the subproblem with the lowest objective value.

\subsection{Preprocessing \label{subsec:preprocessing}}
In this section, we detail a preprocessing algorithm that enables the reduction of size of the GMDTSP instances whose edge costs satisfy the triangle inequality \emph{i.e.,} for distinct $i,j,k \in T$, $c_{ij} + c_{jk} \geq c_{ik}$. A similar algorithm is presented in \cite{Laporte1987, Bektas2011} for the asymmetric generalized traveling salesman problem and generalized vehicle routing problem respectively. In a GMDTSP instance where the edge costs satisfy the triangle inequality, the optimal solution would visit exactly one target in each cluster. We utilize this structure of the optimal solution and reduce the size of a given GMDTSP instance, if possible. To that end, we define a target $i\in T$ to be \emph{dominated} if there exits a target $j \in C_{h(i)}$, $j\neq i$ such that
\begin{enumerate}
\item $c_{pi} + c_{iq} \geq c_{pj} + c_{jq}$ for any $p,q \in T \setminus C_{h(i)}$,
\item $c_{di} \geq c_{dj}$ for all $d \in D$, and
\item $c_{di} + c_{ip} \geq c_{dj} + c_{jp}$ for any $d \in D, p\in T\setminus C_{h(i)}$.
\end{enumerate}
\begin{proposition} \label{prop:pp}
If a dominated target is removed from a GMDTSP instance satisfying triangle inequality, then the optimal cost to the instance does not change. 
\end{proposition}
\begin{proof}
Let $i\in T$ be a dominated vertex. If the target $i$ is not visited in the optimal solution, then its removal does not change the optimal cost. So, assume that $i\in T$ is visited by the optimal solution. Since the edge costs of the instance satisfy the triangle inequality, exactly one target in each cluster is visited by the optimal solution. We now claim that it is possible to exchange the target $i$ with a target $j \in C_{h(i)}$ without increasing the cost of the optimal solution. This follows from the definition of a dominated target. \qed
\end{proof}
The preprocessing checks if a target is dominated and removes the target if it is found so. Then the other targets are checked for dominance relative to the reduced instance. The time complexity of the algorithm is $O(|T|^5)$. 

\subsection{LP rounding heuristic \label{subsec:heuristic}}
We discuss an \emph{LP-rounding} heuristic that aides to generate feasible solutions at the root node and to speed up the convergence of the branch-and-cut algorithm. The heuristic constructs a feasible GMDTSP solution from a given fractional LP solution. It is used only at the root node of the enumeration tree. The heuristic is based on a transformation method in \cite{Oberlin2009}. We are given $\mathbf{y}^*$, the vector of fractional $y_i$ values (denoted by $y_i^f$) for each target $i$. The algorithm proceeds as follows: for each cluster $C_k$ and every target $i\in C_k$, the heuristic sets the value of $y_i$ to $0$ or $1$ according to the condition $y_i^f \geq 0.5$ or $y_i^f < 0.5$ respectively. If every target $i \in C_k$ has $y_i^f < 0.5$, then we set the value of $y_j =1$ where $j = \operatorname{argmax}\{y_i^f:i\in C_k\}$. Once we have assigned the $y_i$ value for each target $i$, we define the set $\Pi := \{i\in T: y_i = 1\}$. We then solve a multiple depot traveling salesman problem (MDTSP) on the set of vertices $\Pi \cup D$. A heuristic based on the transformation method in \cite{Oberlin2009} and LKH heuristic (see \cite{Helsgaun2000}) is used to solve the MDTSP.

\section{Computational results \label{sec:results}}
In this section, we discuss the computational results of the branch-and-cut algorithm. The algorithm was implemented in C++ (gcc version 4.6.3), using the elements of Standard Template Library (STL) in the CPLEX 12.4 framework. As mentioned in Sec. \ref{sec:bandc}, the internal CPLEX cut generation was disabled, and CPLEX was used only to mange the enumeration tree. All the simulations were performed on a Dell Precision T5500 workstation (Inter Xeon E5360 processor @2.53 GHz, 12 GB RAM). The computation times reported are expressed in seconds, and we imposed a time limit of 7200 seconds for each run of the algorithm. The performance of the algorithm was tested on a total of 116 instances, all of which were generated using the generalized traveling salesman problem library (see \cite{Fischetti1997, Gutin2010}). 

\subsection{Problem instances \label{subsec:instance}}
All the computational experiments were conducted on a class of 116 test instances generated from 29 GTSP instances. The GTSP instances are taken directly from the GTSP Instances Library (see \cite{Gutin2010}). The instances are available at \url{http://www.cs.nott.ac.uk/~dxk/gtsp.html}.  For each of the 29 instances, GMDTSP instances with $|D| \in \{2,3,4,5\}$ were generated by assuming the first $|D|$ targets in a GTSP instance to be the set of depots; these depots were then removed from the target clusters. The number of targets in the instances varied from 14 to 105, and the maximum number of target clusters was 21. Hence we had 4 GMDTSP instances for each of the 29 GTSP instances totalling to 116 test instances. We also note that for 64/116 instances, the edge costs do not satisfy the triangle inequality and for the remaining 52 instances, the edge costs satisfy the triangle inequality. The name of the generated instances are the same but for a small modification to spell out the number of depots in the instances. The naming conforms to the format \texttt{GTSPinstancename-D}, where \texttt{GTSPinstancename} corresponds to the GTSP instance name from the library (the first and the last integer in the name corresponds to the number of clusters and the number of targets in the GTSP instance respectively) and \texttt{D} corresponds the number of depots in the instance.

The results are tabulated in tables \ref{tab:results} and \ref{tab:times}. The following nomenclature is used in the table \ref{tab:results}

\noindent \textbf{name}: problem instance name (format: \texttt{GTSPinstancename-D});\\
\textbf{opt}: optimal objective value;\\
\textbf{LB}: objective value of the LP relaxation computed at the root node of the enumeration tree;\\
\textbf{\%LB}: percentage LB/opt;\\
\textbf{UB}: cost of the best feasible solution generated by the LP-rounding heuristic generated at the root node of the enumeration tree;\\
\textbf{\%UB}: percentage UB/opt;\\
\textbf{sec1}: total number of constraints \eqref{eq:sec1} generated;\\
\textbf{sec2}: total number of constraints \eqref{eq:sec2} generated;\\
\textbf{4pec}: total number of constraints \eqref{eq:4path} generated;\\
\textbf{pec}: total number of constraints \eqref{eq:pec1} and \eqref{eq:pec2} generated;\\
\textbf{comb}: total number of constraints \eqref{eq:combb}, \eqref{eq:Tcombb}, and \eqref{eq:Tcombc} generated;\\
\textbf{nodes}: total number of nodes examined in the enumeration tree.\\

The table \ref{tab:times} gives the computational time for each separation routine and the overall the branch-and-cut algorithm. The nomenclature used in table \ref{tab:times} are as follows:

\noindent \textbf{name}: problem instance name (format: \texttt{GTSPinstancename-D});\\
\textbf{total-t}: CPU time, in seconds, for the overall execution of the branch-and-cut algorithm;\\
\textbf{sep-t}: overall CPU time, in seconds, spent for separation;\\
\textbf{sec-t}: CPU time, in seconds, spent for the separation of constraints \eqref{eq:sec1} and \eqref{eq:sec2};\\
\textbf{4pec-t}: CPU time, in seconds, spent for the separation of constraints \eqref{eq:4path};\\
\textbf{pec-t}: CPU time, in seconds, spent for the separation of constraints \eqref{eq:pec1} and \eqref{eq:pec2};\\
\textbf{comb-t}: CPU time, in seconds, spent for the separation of constraints \eqref{eq:combb}, \eqref{eq:Tcombb}, and \eqref{eq:Tcombc};\\
\textbf{\%pec}: percentage of separation time spent for the separation of path elimination constraints \eqref{eq:pec1} and \eqref{eq:pec2}. \\

{\scriptsize
\begin{longtable}{lrrrrrR{0.7cm}R{0.7cm}R{0.7cm}R{0.8cm}R{0.7cm}r}
\caption{Branch-and-cut statistics} \label{tab:results} \\

\toprule
name & opt & LB & \%LB & UB & \%UB & sec1 & sec2 & 4pec & pec & comb & nodes\tabularnewline \midrule
\endfirsthead

\multicolumn{12}{l}
{{\tablename\ \thetable{} -- continued from previous page}} \\
\toprule
name & opt & LB & \%LB & UB & \%UB & sec1 & sec2 & 4pec & pec & comb & nodes\tabularnewline \midrule
\endhead

\hline
\endfoot

\hline \noalign{\vskip1mm} 
\multicolumn{12}{l}{$^\dagger$optimality was not verified within a time-limit of 7200 seconds.}
\endlastfoot

3burma14-2 & 1939 & 1939.00 & 100.00 & 1939 & 100.00 & 51 & 8 & 0 & 2 & 0 & 0\tabularnewline
3burma14-3 & 1664 & 1664.00 & 100.00 & 1664 & 100.00 & 11 & 15 & 0 & 2 & 0 & 0\tabularnewline
3burma14-4 & 1296 & 1296.00 & 100.00 & 1296 & 100.00 & 8 & 14 & 0 & 0 & 0 & 0\tabularnewline
3burma14-5 & 562 & 562.00 & 100.00 & 562 & 100.00 & 1 & 20 & 0 & 0 & 0 & 0\tabularnewline
4br17-2 & 31 & 31.00 & 100.00 & 54 & 174.19 & 7 & 4 & 0 & 0 & 1 & 3\tabularnewline
4br17-3 & 31 & 31.00 & 100.00 & 31 & 100.00 & 7 & 7 & 0 & 0 & 0 & 0\tabularnewline
4br17-4 & 19 & 19.00 & 100.00 & 19 & 100.00 & 5 & 14 & 0 & 0 & 0 & 0\tabularnewline
4br17-5 & 19 & 19.00 & 100.00 & 19 & 100.00 & 5 & 20 & 0 & 4 & 0 & 0\tabularnewline
4gr17-2 & 958 & 846.33 & 88.34 & 965 & 100.73 & 22 & 187 & 8 & 335 & 0 & 97\tabularnewline
4gr17-3 & 738 & 722.88 & 97.95 & 794 & 107.59 & 3 & 43 & 1 & 53 & 4 & 6\tabularnewline
4gr17-4 & 611 & 611.00 & 100.00 & 611 & 100.00 & 2 & 14 & 0 & 3 & 0 & 0\tabularnewline
4gr17-5 & 513 & 513.00 & 100.00 & 513 & 100.00 & 1 & 25 & 0 & 0 & 0 & 0\tabularnewline
4ulysses16-2 & 4695 & 4695.00 & 100.00 & 4695 & 100.00 & 36 & 18 & 0 & 0 & 0 & 0\tabularnewline
4ulysses16-3 & 4695 & 4695.00 & 100.00 & 4695 & 100.00 & 53 & 20 & 0 & 0 & 0 & 0\tabularnewline
4ulysses16-4 & 4695 & 4695.00 & 100.00 & 4695 & 100.00 & 50 & 27 & 0 & 0 & 0 & 0\tabularnewline
4ulysses16-5 & 3914 & 3884.00 & 99.23 & 4188 & 107.00 & 22 & 27 & 0 & 7 & 0 & 3\tabularnewline
5gr21-2 & 1679 & 1531.67 & 91.22 & 1985 & 118.23 & 419 & 367 & 12 & 2158 & 0 & 449\tabularnewline
5gr21-3 & 1024 & 1024.00 & 100.00 & 1024 & 100.00 & 6 & 32 & 0 & 2 & 0 & 0\tabularnewline
5gr21-4 & 953 & 953.00 & 100.00 & 953 & 100.00 & 9 & 20 & 0 & 1 & 0 & 0\tabularnewline
5gr21-5 & 780 & 780.00 & 100.00 & 780 & 100.00 & 4 & 9 & 0 & 2 & 0 & 0\tabularnewline
5gr24-2 & 377 & 340.53 & 90.33 & 828 & 219.63 & 25 & 169 & 0 & 366 & 0 & 13\tabularnewline
5gr24-3 & 377 & 318.00 & 84.35 & 569 & 150.93 & 37 & 181 & 0 & 524 & 32 & 42\tabularnewline
5gr24-4 & 371 & 325.17 & 87.65 & 753 & 202.96 & 39 & 157 & 8 & 303 & 6 & 26\tabularnewline
5gr24-5 & 362 & 308.17 & 85.13 & 739 & 204.14 & 12 & 99 & 7 & 222 & 0 & 87\tabularnewline
5ulysses22-2 & 5199 & 5199.00 & 100.00 & 5199 & 100.00 & 70 & 71 & 2 & 126 & 1 & 0\tabularnewline
5ulysses22-3 & 5311 & 5310.50 & 99.99 & 5442 & 102.47 & 45 & 82 & 0 & 1 & 0 & 3\tabularnewline
5ulysses22-4 & 5021 & 5021.00 & 100.00 & 5021 & 100.00 & 45 & 39 & 0 & 0 & 0 & 0\tabularnewline
5ulysses22-5 & 3913 & 3913.00 & 100.00 & 3913 & 100.00 & 37 & 27 & 0 & 1 & 0 & 0\tabularnewline
6bayg29-2 & 711 & 624.50 & 87.83 & 905 & 127.29 & 82 & 312 & 0 & 1526 & 0 & 148\tabularnewline
6bayg29-3 & 684 & 582.50 & 85.16 & 841 & 122.95 & 70 & 809 & 3 & 3489 & 28 & 301\tabularnewline
6bayg29-4 & 583 & 527.50 & 90.48 & 811 & 139.11 & 25 & 91 & 0 & 171 & 7 & 24\tabularnewline
6bayg29-5 & 565 & 520.79 & 92.17 & 1888 & 334.16 & 40 & 103 & 0 & 360 & 6 & 21\tabularnewline
6bays29-2 & 849 & 761.46 & 89.69 & 1194 & 140.64 & 123 & 178 & 0 & 1466 & 0 & 296\tabularnewline
6bays29-3 & 830 & 777.68 & 93.70 & 1092 & 131.57 & 80 & 145 & 1 & 959 & 17 & 48\tabularnewline
6bays29-4 & 691 & 650.60 & 94.15 & 847 & 122.58 & 30 & 92 & 3 & 238 & 20 & 6\tabularnewline
6bays29-5 & 622 & 591.55 & 95.10 & 1052 & 169.13 & 30 & 99 & 1 & 258 & 3 & 10\tabularnewline
6fri26-2 & 480 & 471.50 & 98.23 & 541 & 112.71 & 54 & 184 & 1 & 519 & 0 & 15\tabularnewline
6fri26-3 & 486 & 466.00 & 95.88 & 510 & 104.94 & 167 & 166 & 0 & 1923 & 3 & 388\tabularnewline
6fri26-4 & 440 & 414.57 & 94.22 & 446 & 101.36 & 92 & 128 & 0 & 355 & 9 & 38\tabularnewline
6fri26-5 & 436 & 411.56 & 94.39 & 473 & 108.49 & 66 & 91 & 2 & 520 & 2 & 41\tabularnewline
9dantzig42-2 & 413 & 413.00 & 100.00 & 413 & 100.00 & 114 & 300 & 0 & 0 & 0 & 0\tabularnewline
9dantzig42-3 & 351 & 351.00 & 100.00 & 358 & 101.99 & 82 & 328 & 0 & 10 & 1 & 3\tabularnewline
9dantzig42-4 & 350 & 345.75 & 98.79 & 396 & 113.14 & 81 & 272 & 1 & 442 & 33 & 6\tabularnewline
9dantzig42-5 & 348 & 344.29 & 98.93 & 348 & 100.00 & 82 & 203 & 2 & 346 & 45 & 12\tabularnewline
10att48-2 & 4924 & 4284.05 & 87.00 & 5510 & 111.90 & 456 & 945 & 0 & 7563 & 0 & 268\tabularnewline
10att48-3 & 4913 & 4539.33 & 92.39 & 6054 & 123.22 & 177 & 880 & 8 & 10115 & 154 & 1406\tabularnewline
10att48-4 & 4428 & 3980.11 & 89.89 & 5685 & 128.39 & 197 & 738 & 2 & 8555 & 138 & 879\tabularnewline
10att48-5 & 4204 & 3897.97 & 92.72 & 5515 & 131.18 & 87 & 690 & 9 & 12826 & 1077 & 594\tabularnewline
10gr48-2 & 1708 & 1707.00 & 99.94 & 1708 & 100.00 & 88 & 186 & 1 & 259 & 0 & 2\tabularnewline
10gr48-3 & 1638 & 1628.14 & 99.40 & 2345 & 143.16 & 74 & 220 & 4 & 1011 & 0 & 14\tabularnewline
10gr48-4 & 1645 & 1629.23 & 99.04 & 2197 & 133.56 & 86 & 185 & 0 & 958 & 1 & 33\tabularnewline
10gr48-5 & 1638 & 1471.48 & 89.83 & 2243 & 136.94 & 108 & 405 & 5 & 2163 & 30 & 179\tabularnewline
10hk48-2 & 6401 & 6209.83 & 97.01 & 6753 & 105.50 & 357 & 418 & 7 & 3018 & 0 & 82\tabularnewline
10hk48-3 & 5872 & 5567.49 & 94.81 & 6211 & 105.77 & 234 & 364 & 1 & 2549 & 0 & 75\tabularnewline
10hk48-4 & 5642 & 5044.00 & 89.40 & 6359 & 112.71 & 269 & 474 & 1 & 2370 & 3 & 69\tabularnewline
10hk48-5 & 5641 & 5145.17 & 91.21 & 6702 & 118.81 & 282 & 399 & 0 & 3455 & 14 & 27\tabularnewline
11berlin52-2 & 3500 & 3425.00 & 97.86 & 4010 & 114.57 & 121 & 288 & 0 & 1 & 1 & 17\tabularnewline
11berlin52-3 & 3500 & 3376.17 & 96.46 & 3963 & 113.23 & 142 & 311 & 1 & 753 & 66 & 20\tabularnewline
11berlin52-4 & 3500 & 3280.00 & 93.71 & 3699 & 105.69 & 88 & 241 & 1 & 426 & 3 & 25\tabularnewline
11berlin52-5 & 3500 & 3273.92 & 93.54 & 4169 & 119.11 & 131 & 160 & 0 & 599 & 26 & 26\tabularnewline
11eil51-2 & 175 & 174.50 & 99.71 & 175 & 100.00 & 148 & 522 & 2 & 1071 & 0 & 3\tabularnewline
11eil51-3 & 174 & 168.83 & 97.03 & 174 & 100.00 & 138 & 269 & 3 & 1160 & 54 & 11\tabularnewline
11eil51-4 & 175 & 165.24 & 94.42 & 183 & 104.57 & 175 & 273 & 11 & 1837 & 18 & 74\tabularnewline
11eil51-5 & 170 & 166.44 & 97.91 & 170 & 100.00 & 71 & 214 & 2 & 479 & 6 & 8\tabularnewline
12brazil58-2 & 14939 & 14939.00 & 100.00 & 14939 & 100.00 & 141 & 278 & 3 & 834 & 0 & 0\tabularnewline
12brazil58-3 & 14930 & 14840.50 & 99.40 & 15240 & 102.08 & 140 & 298 & 1 & 967 & 57 & 18\tabularnewline
12brazil58-4 & 13082 & 12680.46 & 96.93 & 16148 & 123.44 & 147 & 397 & 1 & 1447 & 126 & 40\tabularnewline
12brazil58-5 & 12613 & 11958.93 & 94.81 & 15546 & 123.25 & 153 & 1049 & 1 & 583 & 50 & 98\tabularnewline
14st70-2 & 304 & 288.01 & 94.74 & 307 & 100.99 & 392 & 576 & 2 & 3147 & 3 & 81\tabularnewline
14st70-3 & 301 & 292.57 & 97.20 & 312 & 103.65 & 313 & 600 & 6 & 2846 & 12 & 17\tabularnewline
14st70-4 & 298 & 287.25 & 96.39 & 298 & 100.00 & 182 & 372 & 4 & 1404 & 4 & 19\tabularnewline
14st70-5 & 298 & 282.28 & 94.73 & 325 & 109.06 & 313 & 670 & 9 & 3883 & 5 & 163\tabularnewline
16eil76-2 & 198 & 198.00 & 100.00 & 198 & 100.00 & 223 & 436 & 0 & 945 & 0 & 0\tabularnewline
16eil76-3 & 197 & 197.00 & 100.00 & 197 & 100.00 & 174 & 258 & 3 & 727 & 6 & 0\tabularnewline
16eil76-4 & 197 & 197.00 & 100.00 & 197 & 100.00 & 147 & 360 & 4 & 941 & 20 & 0\tabularnewline
16eil76-5 & 188 & 180.42 & 95.97 & 196 & 104.26 & 233 & 386 & 5 & 1132 & 25 & 27\tabularnewline
20gr96-2$^\dagger$ & 29966 & 28357.03 & 94.63 & 30821 & 102.85 & 823 & 1220 & 1 & 3540 & 0 & 62\tabularnewline
20gr96-3$^\dagger$ & 29621 & 29263.93 & 98.79 & 30768 & 103.87 & 876 & 1326 & 2 & 3382 & 529 & 50\tabularnewline
20gr96-4 & 28705 & 27650.63 & 96.33 & 30121 & 104.93 & 866 & 1754 & 6 & 4268 & 7 & 144\tabularnewline
20gr96-5 & 28598 & 27768.50 & 97.10 & 29976 & 104.82 & 676 & 1269 & 1 & 2087 & 1 & 52\tabularnewline
20kroA100-2 & 9630 & 9265.75 & 96.22 & 9769 & 101.44 & 746 & 1080 & 5 & 3481 & 0 & 66\tabularnewline
20kroA100-3 & 9334 & 8935.25 & 95.73 & 9535 & 102.15 & 532 & 915 & 0 & 2801 & 0 & 92\tabularnewline
20kroA100-4 & 8897 & 8539.03 & 95.98 & 10243 & 115.13 & 935 & 1241 & 2 & 4490 & 0 & 126\tabularnewline
20kroA100-5 & 8827 & 8477.39 & 96.04 & 9020 & 102.19 & 520 & 1028 & 4 & 2480 & 0 & 47\tabularnewline
20kroB100-2 & 9800 & 9492.00 & 96.86 & 10382 & 105.94 & 510 & 955 & 4 & 3025 & 0 & 30\tabularnewline
20kroB100-3$^\dagger$ & 10218 & 9197.41 & 90.01 & 10300 & 100.80 & 903 & 1120 & 1 & 5373 & 0 & 130\tabularnewline
20kroB100-4 & 9564 & 9293.31 & 97.17 & 9637 & 100.76 & 361 & 714 & 0 & 2323 & 0 & 20\tabularnewline
20kroB100-5 & 9226 & 8525.71 & 92.41 & 11708 & 126.90 & 739 & 1058 & 10 & 7225 & 0 & 119\tabularnewline
20kroC100-2$^\dagger$ & 10089 & 9548.13 & 94.64 & 10089 & 100.00 & 420 & 974 & 0 & 1551 & 0 & 3\tabularnewline
20kroC100-3 & 9244 & 9130.82 & 98.78 & 9346 & 101.10 & 494 & 1006 & 0 & 1940 & 1 & 8\tabularnewline
20kroC100-4 & 9292 & 9061.20 & 97.52 & 9342 & 100.54 & 307 & 707 & 2 & 1132 & 3 & 10\tabularnewline
20kroC100-5 & 9252 & 8991.89 & 97.19 & 10437 & 112.81 & 380 & 956 & 3 & 2181 & 0 & 19\tabularnewline
20kroD100-2$^\dagger$ & 9353 & 8497.63 & 90.85 & 9381 & 100.30 & 886 & 1525 & 4 & 3221 & 6 & 65\tabularnewline
20kroD100-3 & 8813 & 8130.12 & 92.25 & 11404 & 129.40 & 1284 & 1664 & 5 & 11642 & 24 & 212\tabularnewline
20kroD100-4 & 8772 & 8283.74 & 94.43 & 8823 & 100.58 & 577 & 1067 & 11 & 3230 & 3 & 67\tabularnewline
20kroD100-5 & 8677 & 8233.85 & 94.89 & 9247 & 106.57 & 478 & 732 & 1 & 3277 & 0 & 45\tabularnewline
20kroE100-2 & 9526 & 9290.65 & 97.53 & 10207 & 107.15 & 599 & 1098 & 7 & 4461 & 0 & 45\tabularnewline
20kroE100-3 & 9262 & 9153.61 & 98.83 & 9854 & 106.39 & 612 & 1048 & 7 & 3974 & 19 & 26\tabularnewline
20kroE100-4 & 9262 & 9147.56 & 98.76 & 11046 & 119.26 & 513 & 1032 & 3 & 3410 & 4 & 21\tabularnewline
20kroE100-5 & 9081 & 8900.07 & 98.01 & 9707 & 106.89 & 391 & 925 & 3 & 2802 & 0 & 32\tabularnewline
20rat99-2 & 505 & 504.33 & 99.87 & 521 & 103.17 & 507 & 951 & 0 & 0 & 0 & 7\tabularnewline
20rat99-3 & 504 & 498.23 & 98.85 & 543 & 107.74 & 528 & 977 & 4 & 1582 & 1 & 20\tabularnewline
20rat99-4 & 501 & 490.67 & 97.94 & 515 & 102.79 & 958 & 1259 & 5 & 10214 & 0 & 2383\tabularnewline
20rat99-5 & 487 & 477.67 & 98.08 & 506 & 103.90 & 688 & 967 & 4 & 4320 & 0 & 376\tabularnewline
20rd100-2$^\dagger$ & 3459 & 3380.39 & 97.73 & 3714 & 107.37 & 742 & 1406 & 0 & 4119 & 0 & 42\tabularnewline
20rd100-3 & 3383 & 3218.89 & 95.15 & 3384 & 100.03 & 657 & 1456 & 2 & 4238 & 1 & 55\tabularnewline
20rd100-4 & 3298 & 3167.38 & 96.04 & 3398 & 103.03 & 530 & 889 & 2 & 2651 & 0 & 29\tabularnewline
20rd100-5 & 3234 & 3109.99 & 96.17 & 3327 & 102.88 & 559 & 1056 & 6 & 4114 & 1 & 64\tabularnewline
21eil101-2 & 248 & 245.41 & 98.96 & 255 & 102.82 & 387 & 812 & 0 & 1476 & 0 & 20\tabularnewline
21eil101-3 & 248 & 243.04 & 98.00 & 267 & 107.66 & 570 & 982 & 4 & 2371 & 6 & 37\tabularnewline
21eil101-4 & 233 & 230.2759 & 98.83 & 251 & 107.73 & 432 & 629 & 3 & 2586 & 0 & 15\tabularnewline
21eil101-5 & 232 & 226.33 & 97.56 & 257 & 110.78 & 275 & 527 & 0 & 1483 & 2 & 16\tabularnewline
21lin105-2 & 8358 & 8316.43 & 99.50 & 8726 & 104.40 & 652 & 1122 & 0 & 0 & 0 & 16\tabularnewline
21lin105-3$^\dagger$ & 8304 & 8164.21 & 98.32 & 8619 & 103.79 & 870 & 1298 & 3 & 25572 & 22 & 7103\tabularnewline
21lin105-4 & 7827 & 7695.17 & 98.32 & 8365 & 106.87 & 619 & 941 & 2 & 888 & 12 & 89\tabularnewline
21lin105-5$^\dagger$ & 8052 & 7568.64 & 94.00 & 8110 & 100.72 & 745 & 1166 & 1 & 2419 & 6 & 145\tabularnewline
\end{longtable}}

{\scriptsize
\begin{longtable}{lR{1.1cm}R{1cm}R{1cm}R{1cm}R{1cm}R{1.2cm}r}
\caption{Algorithm computation times} \label{tab:times} \\

\toprule
name & total-t & sep-t & sec-t & 4pec-t & pec-t & comb-t & \%pec\tabularnewline \midrule
\endfirsthead

\multicolumn{8}{l}
{{\tablename\ \thetable{} -- continued from previous page}} \\
\toprule
name & total-t & sep-t & sec-t & 4pec-t & pec-t & comb-t & \%pec\tabularnewline \midrule
\endhead

\hline
\endfoot

\hline \noalign{\vskip1mm} 
\multicolumn{8}{l}{$^\dagger$optimality was not verified within a time-limit of 7200 seconds.}
\endlastfoot

3burma14-2 & 0.07 & 0.00 & 0.00 & 0.00 & 0.00 & 0.00 & 3.13\tabularnewline
3burma14-3 & 0.02 & 0.00 & 0.00 & 0.00 & 0.00 & 0.00 & 2.68\tabularnewline
3burma14-4 & 0.02 & 0.00 & 0.00 & 0.00 & 0.00 & 0.00 & 1.97\tabularnewline
3burma14-5 & 0.02 & 0.00 & 0.00 & 0.00 & 0.00 & 0.00 & 3.50\tabularnewline
4br17-2 & 0.03 & 0.00 & 0.00 & 0.00 & 0.00 & 0.00 & 1.14\tabularnewline
4br17-3 & 0.01 & 0.00 & 0.00 & 0.00 & 0.00 & 0.00 & 0.00\tabularnewline
4br17-4 & 0.02 & 0.00 & 0.00 & 0.00 & 0.00 & 0.00 & 0.00\tabularnewline
4br17-5 & 0.04 & 0.01 & 0.00 & 0.00 & 0.00 & 0.00 & 68.52\tabularnewline
4gr17-2 & 1.16 & 0.33 & 0.10 & 0.00 & 0.22 & 0.01 & 65.71\tabularnewline
4gr17-3 & 0.23 & 0.05 & 0.01 & 0.00 & 0.04 & 0.00 & 74.03\tabularnewline
4gr17-4 & 0.02 & 0.00 & 0.00 & 0.00 & 0.00 & 0.00 & 0.00\tabularnewline
4gr17-5 & 0.01 & 0.00 & 0.00 & 0.00 & 0.00 & 0.00 & 0.00\tabularnewline
4ulysses16-2 & 0.05 & 0.00 & 0.00 & 0.00 & 0.00 & 0.00 & 1.71\tabularnewline
4ulysses16-3 & 0.05 & 0.00 & 0.00 & 0.00 & 0.00 & 0.00 & 2.04\tabularnewline
4ulysses16-4 & 0.08 & 0.00 & 0.00 & 0.00 & 0.00 & 0.00 & 1.93\tabularnewline
4ulysses16-5 & 0.13 & 0.02 & 0.01 & 0.00 & 0.02 & 0.00 & 72.63\tabularnewline
5gr21-2 & 12.89 & 3.63 & 1.00 & 0.00 & 2.54 & 0.09 & 69.98\tabularnewline
5gr21-3 & 0.04 & 0.00 & 0.00 & 0.00 & 0.00 & 0.00 & 2.28\tabularnewline
5gr21-4 & 0.02 & 0.00 & 0.00 & 0.00 & 0.00 & 0.00 & 2.86\tabularnewline
5gr21-5 & 0.07 & 0.00 & 0.00 & 0.00 & 0.00 & 0.00 & 2.81\tabularnewline
5gr24-2 & 1.81 & 0.45 & 0.07 & 0.00 & 0.38 & 0.00 & 84.82\tabularnewline
5gr24-3 & 3.51 & 0.92 & 0.18 & 0.00 & 0.73 & 0.01 & 79.17\tabularnewline
5gr24-4 & 2.89 & 0.76 & 0.11 & 0.00 & 0.64 & 0.01 & 83.80\tabularnewline
5gr24-5 & 1.63 & 0.38 & 0.12 & 0.00 & 0.25 & 0.01 & 65.26\tabularnewline
5ulysses22-2 & 0.77 & 0.18 & 0.04 & 0.00 & 0.13 & 0.00 & 74.26\tabularnewline
5ulysses22-3 & 0.43 & 0.03 & 0.03 & 0.00 & 0.00 & 0.00 & 0.64\tabularnewline
5ulysses22-4 & 0.18 & 0.02 & 0.02 & 0.00 & 0.00 & 0.00 & 0.75\tabularnewline
5ulysses22-5 & 0.06 & 0.01 & 0.01 & 0.00 & 0.00 & 0.00 & 1.82\tabularnewline
6bayg29-2 & 18.69 & 4.97 & 0.73 & 0.00 & 4.17 & 0.08 & 83.79\tabularnewline
6bayg29-3 & 20.50 & 5.66 & 1.31 & 0.00 & 4.19 & 0.15 & 74.10\tabularnewline
6bayg29-4 & 1.26 & 0.31 & 0.06 & 0.00 & 0.24 & 0.01 & 77.32\tabularnewline
6bayg29-5 & 1.19 & 0.27 & 0.08 & 0.00 & 0.18 & 0.01 & 68.11\tabularnewline
6bays29-2 & 21.40 & 6.19 & 0.96 & 0.00 & 5.14 & 0.08 & 83.16\tabularnewline
6bays29-3 & 10.60 & 2.78 & 0.33 & 0.00 & 2.43 & 0.02 & 87.50\tabularnewline
6bays29-4 & 1.22 & 0.30 & 0.05 & 0.00 & 0.24 & 0.01 & 80.74\tabularnewline
6bays29-5 & 0.97 & 0.22 & 0.04 & 0.00 & 0.18 & 0.00 & 79.98\tabularnewline
6fri26-2 & 5.55 & 1.34 & 0.12 & 0.00 & 1.22 & 0.01 & 90.53\tabularnewline
6fri26-3 & 18.32 & 5.55 & 1.11 & 0.00 & 4.31 & 0.13 & 77.68\tabularnewline
6fri26-4 & 3.75 & 0.92 & 0.12 & 0.00 & 0.78 & 0.01 & 85.23\tabularnewline
6fri26-5 & 3.26 & 0.83 & 0.12 & 0.00 & 0.70 & 0.01 & 84.67\tabularnewline
9dantzig42-2 & 1.07 & 0.28 & 0.27 & 0.00 & 0.00 & 0.01 & 0.38\tabularnewline
9dantzig42-3 & 1.26 & 0.34 & 0.16 & 0.00 & 0.18 & 0.00 & 51.77\tabularnewline
9dantzig42-4 & 5.15 & 1.29 & 0.22 & 0.00 & 1.05 & 0.01 & 81.81\tabularnewline
9dantzig42-5 & 7.97 & 1.93 & 0.20 & 0.00 & 1.71 & 0.01 & 88.71\tabularnewline
10att48-2 & 280.75 & 80.02 & 6.73 & 0.00 & 72.88 & 0.41 & 91.08\tabularnewline
10att48-3 & 243.27 & 71.62 & 9.29 & 0.00 & 60.66 & 1.67 & 84.70\tabularnewline
10att48-4 & 203.20 & 59.39 & 7.56 & 0.00 & 50.63 & 1.19 & 85.26\tabularnewline
10att48-5 & 130.36 & 38.93 & 5.95 & 0.00 & 31.74 & 1.23 & 81.55\tabularnewline
10gr48-2 & 9.25 & 2.26 & 0.21 & 0.00 & 2.04 & 0.01 & 90.50\tabularnewline
10gr48-3 & 31.81 & 7.87 & 0.54 & 0.00 & 7.30 & 0.03 & 92.72\tabularnewline
10gr48-4 & 39.36 & 9.62 & 0.60 & 0.00 & 8.96 & 0.06 & 93.10\tabularnewline
10gr48-5 & 43.79 & 11.76 & 1.39 & 0.00 & 10.17 & 0.20 & 86.48\tabularnewline
10hk48-2 & 273.81 & 69.58 & 3.29 & 0.00 & 66.15 & 0.14 & 95.07\tabularnewline
10hk48-3 & 170.99 & 43.05 & 1.76 & 0.00 & 41.19 & 0.10 & 95.66\tabularnewline
10hk48-4 & 35.98 & 9.64 & 1.04 & 0.00 & 8.51 & 0.09 & 88.28\tabularnewline
10hk48-5 & 92.75 & 24.49 & 1.57 & 0.00 & 22.84 & 0.08 & 93.27\tabularnewline
11berlin52-2 & 2.28 & 1.06 & 1.03 & 0.00 & 0.00 & 0.02 & 0.37\tabularnewline
11berlin52-3 & 67.95 & 16.48 & 0.95 & 0.00 & 15.48 & 0.05 & 93.91\tabularnewline
11berlin52-4 & 27.96 & 7.19 & 0.44 & 0.00 & 6.72 & 0.04 & 93.41\tabularnewline
11berlin52-5 & 19.57 & 5.17 & 0.46 & 0.00 & 4.66 & 0.05 & 90.16\tabularnewline
11eil51-2 & 200.63 & 48.72 & 1.39 & 0.00 & 47.29 & 0.03 & 97.08\tabularnewline
11eil51-3 & 100.95 & 24.48 & 0.98 & 0.00 & 23.47 & 0.03 & 95.85\tabularnewline
11eil51-4 & 142.50 & 37.00 & 1.94 & 0.00 & 34.95 & 0.11 & 94.45\tabularnewline
11eil51-5 & 33.19 & 8.25 & 0.36 & 0.00 & 7.87 & 0.02 & 95.42\tabularnewline
12brazil58-2 & 33.00 & 7.94 & 0.96 & 0.00 & 6.95 & 0.03 & 87.51\tabularnewline
12brazil58-3 & 56.51 & 13.29 & 0.93 & 0.00 & 12.31 & 0.06 & 92.60\tabularnewline
12brazil58-4 & 32.61 & 8.62 & 1.00 & 0.00 & 7.53 & 0.09 & 87.35\tabularnewline
12brazil58-5 & 3.48 & 1.06 & 0.52 & 0.00 & 0.44 & 0.10 & 41.55\tabularnewline
14st70-2 & 876.36 & 222.60 & 6.73 & 0.00 & 215.47 & 0.39 & 96.80\tabularnewline
14st70-3 & 1071.01 & 264.38 & 4.16 & 0.00 & 260.10 & 0.12 & 98.38\tabularnewline
14st70-4 & 354.16 & 87.56 & 1.86 & 0.00 & 85.61 & 0.08 & 97.78\tabularnewline
14st70-5 & 429.46 & 113.03 & 5.51 & 0.00 & 106.96 & 0.57 & 94.63\tabularnewline
16eil76-2 & 160.97 & 38.04 & 1.72 & 0.00 & 36.27 & 0.04 & 95.36\tabularnewline
16eil76-3 & 71.48 & 17.47 & 0.80 & 0.00 & 16.64 & 0.03 & 95.24\tabularnewline
16eil76-4 & 173.67 & 43.19 & 1.11 & 0.00 & 42.03 & 0.05 & 97.31\tabularnewline
16eil76-5 & 274.12 & 69.50 & 1.87 & 0.00 & 67.52 & 0.12 & 97.15\tabularnewline
20gr96-2$^\dagger$ & 7200.00 & 1901.87 & 44.02 & 0.00 & 1857.29 & 0.56 & 97.66\tabularnewline
20gr96-3$^\dagger$ & 7200.00 & 1862.37 & 38.38 & 0.00 & 1823.22 & 0.77 & 97.90\tabularnewline
20gr96-4 & 5467.42 & 1428.08 & 48.45 & 0.00 & 1378.35 & 1.28 & 96.52\tabularnewline
20gr96-5 & 6495.00 & 1643.50 & 35.00 & 0.00 & 1607.79 & 0.71 & 97.83\tabularnewline
20kroA100-2 & 4291.87 & 1091.52 & 22.62 & 0.00 & 1068.47 & 0.42 & 97.89\tabularnewline
20kroA100-3 & 4225.89 & 1060.29 & 14.82 & 0.00 & 1044.91 & 0.56 & 98.55\tabularnewline
20kroA100-4 & 5057.47 & 1300.82 & 28.19 & 0.00 & 1271.60 & 1.04 & 97.75\tabularnewline
20kroA100-5 & 6368.98 & 1606.81 & 20.13 & 0.00 & 1585.98 & 0.70 & 98.70\tabularnewline
20kroB100-2 & 3389.43 & 841.28 & 12.24 & 0.00 & 828.79 & 0.25 & 98.52\tabularnewline
20kroB100-3$^\dagger$ & 7200.04 & 1838.03 & 33.81 & 0.00 & 1803.14 & 1.08 & 98.10\tabularnewline
20kroB100-4 & 3120.43 & 778.88 & 9.44 & 0.00 & 769.15 & 0.29 & 98.75\tabularnewline
20kroB100-5 & 3397.49 & 883.26 & 24.75 & 0.00 & 857.50 & 1.01 & 97.08\tabularnewline
20kroC100-2$^\dagger$ & 7200.00 & 1821.34 & 15.18 & 0.00 & 1805.91 & 0.25 & 99.15\tabularnewline
20kroC100-3 & 3052.62 & 747.14 & 10.82 & 0.00 & 736.09 & 0.23 & 98.52\tabularnewline
20kroC100-4 & 1009.37 & 250.86 & 4.82 & 0.00 & 245.88 & 0.16 & 98.01\tabularnewline
20kroC100-5 & 2839.31 & 713.70 & 11.93 & 0.00 & 701.39 & 0.38 & 98.28\tabularnewline
20kroD100-2$^\dagger$ & 7200.00 & 1852.91 & 33.91 & 0.00 & 1818.46 & 0.54 & 98.14\tabularnewline
20kroD100-3 & 6287.9 & 1671.43 & 50.47 & 0.00 & 1619.66 & 1.30 & 96.90\tabularnewline
20kroD100-4 & 4716.98 & 1190.26 & 18.79 & 0.00 & 1170.92 & 0.55 & 98.38\tabularnewline
20kroD100-5 & 2669.25 & 671.32 & 13.10 & 0.00 & 657.78 & 0.44 & 97.98\tabularnewline
20kroE100-2 & 4718.14 & 1204.19 & 24.14 & 0.00 & 1179.63 & 0.41 & 97.96\tabularnewline
20kroE100-3 & 4737.91 & 1147.37 & 24.29 & 0.00 & 1122.59 & 0.49 & 97.84\tabularnewline
20kroE100-4 & 2624.53 & 641.08 & 17.04 & 0.00 & 623.69 & 0.35 & 97.29\tabularnewline
20kroE100-5 & 1892.52 & 476.91 & 10.32 & 0.00 & 466.24 & 0.35 & 97.76\tabularnewline
20rat99-2 & 65.57 & 12.65 & 12.55 & 0.00 & 0.02 & 0.09 & 0.15\tabularnewline
20rat99-3 & 2416.98 & 583.46 & 14.15 & 0.00 & 569.01 & 0.30 & 97.52\tabularnewline
20rat99-4 & 6091.56 & 1414.13 & 140.03 & 0.00 & 1245.85 & 28.26 & 88.10\tabularnewline
20rat99-5 & 3165.79 & 747.76 & 46.84 & 0.00 & 693.47 & 7.45 & 92.74\tabularnewline
20rd100-2$^\dagger$ & 7200.00 & 1846.05 & 37.12 & 0.00 & 1808.40 & 0.52 & 97.96\tabularnewline
20rd100-3 & 3815.24 & 969.42 & 23.26 & 0.00 & 945.69 & 0.47 & 97.55\tabularnewline
20rd100-4 & 3273.97 & 826.82 & 16.76 & 0.00 & 809.60 & 0.46 & 97.92\tabularnewline
20rd100-5 & 2513.41 & 643.81 & 15.04 & 0.00 & 628.22 & 0.55 & 97.58\tabularnewline
21eil101-2 & 2100.39 & 519.56 & 10.63 & 0.00 & 508.75 & 0.19 & 97.92\tabularnewline
21eil101-3 & 4245.95 & 1069.99 & 18.31 & 0.00 & 1051.25 & 0.43 & 98.25\tabularnewline
21eil101-4 & 906.82 & 227.88 & 7.48 & 0.00 & 220.15 & 0.25 & 96.61\tabularnewline
21eil101-5 & 682.82 & 172.40 & 4.07 & 0.00 & 168.13 & 0.19 & 97.52\tabularnewline
21lin105-2 & 86.33 & 21.14 & 20.93 & 0.00 & 0.03 & 0.18 & 0.15\tabularnewline
21lin105-3$^\dagger$ & 7200.00 & 2047.88 & 380.14 & 0.00 & 1566.72 & 101.02 & 76.50\tabularnewline
21lin105-4 & 3609.22 & 903.74 & 19.51 & 0.00 & 883.49 & 0.74 & 97.76\tabularnewline
21lin105-5$^\dagger$ & 7200.00 & 1890.67 & 45.87 & 0.00 & 1843.24 & 1.56 & 97.49\tabularnewline
\end{longtable}}

The results indicate that the proposed branch-and-cut algorithm can solve instances involving up to 105 targets with modest computation times. The preprocessing algorithm in Sec. \ref{subsec:preprocessing} was applied to 53/116 instances. The time taken by the preprocessing algorithm is not included in the overall computation time. The preprocessing algorithm reduced the size of these instances by 6 targets on average and the maximum reduction obtained was 14 targets.  We observe that the instances that have a larger number of violated path elimination constraints take considerably large amount of computation time. The last column in table \ref{tab:times}, whose average is 73\%, indicates the percentage of separation time spent for finding violated path elimination constraints. This is not surprising because the time complexity for identifying violated path elimination constraints in \eqref{eq:pec1} and \eqref{eq:pec2} given a fractional solution, is $O(|T|^5)$ and $O(m|T|^4)$ respectively.  The average number of T-comb inequalities that were generated in the enumeration tree were larger for some of the bigger instances (see table \ref{tab:results}). They were effective, especially in tightening the lower bound for the instances that were not solved to optimality; for the instances where violated T-comb inequalities were separated out, the average linear programming relaxation gap improvement was 18\%. They were also useful in reducing the computation times for larger instances despite increasing the computation times for smaller instances. Overall, we were able to solve 108/116 instances to the optimality with the largest instance involving 105 targets, 21 clusters and 5 depots. The ``\textbf{opt}'' column for the remaining 6/116 instances is the cost of the best feasible solution obtained by the branch-and-cut algorithm at the end of 7200 seconds. For the instances not solved to optimality within the time limit of 7200 seconds, the LP-rounding heuristic was effective in  generating feasible solutions within 2.1\% of the best feasible solution, on average. 

\section{Conclusion \label{sec:conc}}
In summary, we have presented an exact algorithm for the GMDTSP, a problem that has several practical applications including maritime transportation, health-care logistics, survivable telecommunication network design, and routing unmanned vehicles to name a few. A mixed-integer linear programming formulation including several classes of valid inequalities was proposed the facial structure of the polytope of feasible solutions was studied in detail. All the results were used to develop a branch-and-cut algorithm whose performance was corroborated through extensive numerical experiments on a wide range of benchmark instances from the standard library. The largest solved instance involved 105 targets, 21 clusters and 4 depots. Future work can be directed towards development of branch-and-cut approaches accompanied with a polyhedral study to solve the asymmetric counterpart of the problem. 

\section*{\refname}
\bibliography{GMDTSP}
\bibliographystyle{plainnat}
\end{document}